\documentclass[11pt]{article}

\usepackage[margin=0.8in]{geometry}

\usepackage{varioref}
\usepackage[usenames,svgnames,xcdraw,table]{xcolor}
\definecolor{DarkBlue}{rgb}{0.1,0.1,0.5}
\definecolor{DarkGreen}{rgb}{0.1,0.5,0.1}

\usepackage{nicefrac}

\usepackage{hyperref}
\hypersetup{
   colorlinks   = true,
 linkcolor    = DarkBlue, % color of internal links
 urlcolor     = DarkBlue, % color of external links
	 citecolor    = DarkGreen % color of links to bibliography
}

\usepackage{appendix}

\usepackage{mathtools}

\usepackage{algorithmic}
\usepackage{algorithm}
\usepackage{array}

\usepackage[font={small,it}]{caption}

\usepackage{graphicx,epsfig,amsmath,latexsym,amssymb,verbatim}%,revsymb}

\usepackage[square,sort,semicolon,numbers]{natbib}%[round]

\newcommand{\extra}[1]{}

\usepackage{amsmath,amssymb,amsfonts}
\usepackage{amsthm}
\usepackage{thmtools,thm-restate}
\usepackage{enumitem}

\newtheorem{theorem}{Theorem}

\newtheorem{lemma}[theorem]{Lemma}

\newtheorem{claim}{Claim}

\def\squareforqed{\hbox{\rlap{$\sqcap$}$\sqcup$}}
\def\qed{\ifmmode\squareforqed\else{\unskip\nobreak\hfil
\penalty50\hskip1em\null\nobreak\hfil\squareforqed
\parfillskip=0pt\finalhyphendemerits=0\endgraf}\fi}
\def\endenv{\ifmmode\;\else{\unskip\nobreak\hfil
\penalty50\hskip1em\null\nobreak\hfil\;
\parfillskip=0pt\finalhyphendemerits=0\endgraf}\fi}

\renewenvironment{proof}{\noindent \textbf{{Proof~} }}{\qed\medskip}
\newenvironment{proof+}[1]{\noindent \textbf{{Proof #1~} }}{\qed\medskip}

\mathchardef\ordinarycolon\mathcode`\:
\mathcode`\:=\string"8000
\def\vcentcolon{\mathrel{\mathop\ordinarycolon}}
\begingroup \catcode`\:=\active
  \lowercase{\endgroup
  \let :\vcentcolon
  }

\newcommand{\EFX}{\textsc{EFX}}
\newcommand{\NSW}{\mathrm{NSW}}

\DeclareMathOperator*{\argmax}{arg\,max}

\usepackage{amsmath}
\usepackage{amssymb,amsthm,mathtools}
\usepackage{amsmath,amsfonts} % Math packages
\usepackage{caption} % Custom captions for tables etc.
\usepackage{xcolor}
\usepackage{hyperref}
\usepackage{comment}

\usepackage{verbatim}

\newcommand{\goods}{\mathcal{G}}

\newcommand{\M}{{\rm M}}
\newcommand{\EFone}{\textsc{Ef1}}
\newcommand{\PROPone}{\textsc{Prop1}}

\usepackage[font=small,labelfont=bf]{caption}
\usepackage{xspace}
\usepackage{etoolbox}

\title{\bfseries Tight Approximation Algorithms for $p$-Mean Welfare \\ Under Subadditive Valuations}
\author{Siddharth Barman\thanks{Indian Institute of Science. {\tt barman@iisc.ac.in}} \and Umang Bhaskar\thanks{Tata Institute of Fundamental Research. {\tt umang@tifr.res.in}} \and Anand Krishna\thanks{Indian Institute of Science. {\tt anandkrishna@iisc.ac.in}} \and Ranjani G.~Sundaram\thanks{Chennai Mathematical Institute. {\tt ranjanigs@cmi.ac.in}}}

\date{}

\begin{document}
\maketitle

\begin{abstract}
We develop polynomial-time algorithms for the fair and efficient allocation of indivisible goods among $n$ agents that have subadditive valuations over the goods. We first consider the Nash social welfare as our  
objective and design a polynomial-time algorithm that, in the value oracle model, finds an $8n$-approximation to the Nash optimal allocation. Subadditive valuations include XOS (fractionally subadditive) and submodular valuations as special cases. Our result, even for the special case of 
submodular valuations, improves upon the previously best known $O(n \log n)$-approximation ratio of Garg et al. (2020). 

More generally, we study maximization of $p$-mean welfare. The $p$-mean welfare is parameterized by an exponent term $p \in (-\infty, 1]$ and encompasses a range of welfare functions, such as social welfare ($p = 1$), Nash social welfare ($p \to 0$), and egalitarian welfare ($p \to -\infty$). We give an algorithm that, for subadditive valuations and any given $p \in (-\infty, 1]$, computes (in the value oracle model and in polynomial time) an allocation with $p$-mean welfare at least $\nicefrac{1}{8n}$ times the optimal. 

Further, we show that our approximation guarantees are essentially tight for XOS and, hence, subadditive valuations. We adapt a result of Dobzinski et al. (2010) to show that, under XOS valuations, an $O \left(n^{1-\varepsilon} \right)$ approximation for the $p$-mean welfare for any $p \in (-\infty,1]$ (including the Nash social welfare) requires exponentially many value queries; here, $\varepsilon>0$ is any fixed constant.
\end{abstract}

\section{Introduction}

%How should a set $\goods$ of $m$ indivisible goods be divided between a set $\agents$ of $n$ agents? 
In discrete fair division, given a set of $m$ goods and $n$ agents, the problem is to integrally allocate the set of goods to the agents in a fair and (economically) efficient manner~\cite{brandt2016handbook,endriss2017trends,aziz2019developments}. In this thread of work, the Nash social welfare---defined as the geometric mean of the agents' valuations for their assigned bundles---has emerged as a fundamental and prominent measure of the quality of an allocation. %It is an important measure of the collective welfare of an allocation, and 
It provides a balance between two central objectives: the social welfare (the sum of the agents' valuations) and the egalitarian welfare (the minimum valuation across the agents). Note that social welfare is a standard measure of (economic) efficiency, whereas egalitarian welfare is a fairness objective.

A Nash optimal allocation (i.e., an allocation that maximizes Nash social welfare) satisfies other fairness and efficiency criteria as well.  
%Let \emph{optimal allocation} refer to the allocation that maximizes the Nash social welfare. 
Such an allocation is clearly Pareto optimal. Furthermore, if agents have additive valuations, then a Nash optimal allocation is known to be fair in the sense that it is guaranteed to be \emph{envy-free up to one good} (\EFone)~\cite{CaragiannisKMPS19} and \emph{proportional up to one good} (\PROPone)~\cite{ConitzerFS17}.\footnote{An allocation is said to be $\EFone$ if for any pair of agents $i$ and $j$, there exists a good $g$ in $j$'s bundle, such that $i$ prefers her bundle to the one obtained after removing $g$ from $j$'s bundle. An allocation is said to be $\PROPone$ if for each agent $i$ there exists a good $g$ with the property that including $g$ into $i$'s bundle ensures that $i$ achieves a proportional share, i.e., her valuation ends up being at least $1/n$ times her value for all the goods.} %$\EFone$ is also a well-studied notion of fairness in discrete fair division~\cite{aziz2019developments}. 

As an objective, Nash social welfare is scale invariant: multiplicatively scaling any agent's valuation function by a nonnegative factor does not change the Nash optimal allocation. Furthermore, interesting connections have been established between market models and this welfare function; see, e.g.,~\cite{ColeDGJMVY17,BarmanKV18}. As a practical application, the website {\tt spliddit.org} uses the Nash social welfare as the optimization objective when partitioning indivisible goods~\cite{goldman2015spliddit,CaragiannisKMPS19}.

However, computing a Nash optimal allocation is {\rm APX}-hard, even when the agents have additive valuations~\cite{Lee17}. In terms of approximation algorithms, the problem of maximizing Nash social welfare has received considerable attention in recent years \cite{CG15approximating,ColeDGJMVY17,AGS+17nash, BGH+17earning,AGM+18nash,BarmanKV18,GHM18Approximating}. In particular, a polynomial-time $e^{1/e}$-approximation algorithm is known for additive valuations \cite{BarmanKV18}. This algorithm preserves $\EFone$, up to a factor of $(1+\varepsilon)$, and Pareto optimality. The approximation guarantee of $ e^{1/e} \approx 1.45$ also holds for budget-additive valuations~\cite{ChaudhuryCGGHM18}. 
%Recently, the computation of the optimal Nash social welfare for submodular utilities is studied by Garg et al.~\cite{GargKK20}. 
The work of Garg et al.~\cite{GargKK20} extends this line of work by considering Nash social welfare maximization under submodular valuations. 

Submodular valuations capture the diminishing marginal returns property. They constitute a subclass of subadditive valuations, which, in turn, model complement-freeness. Formally, a set function $v$ (defined over a set of indivisible goods) is subadditive if it satisfies $v(A \cup B) \leq v(A) + v(B)$, for all subsets of goods $A$ and $B$. Complement-freeness is a very common assumption on valuation functions. Hence, fair division with subadditive valuations is an encompassing and important problem. 
% considered in computer science and mathematical economics

For Nash social welfare maximization under submodular valuations, Garg et al.~\cite{GargKK20} obtain an $O(n \log n)$-approximation algorithm. Prior to their work, the best known approximation ratio for submodular valuations was $(m-n+1)$, which also extends to subadditive valuations~\cite{NR14}; here, $m$ denotes the number of  goods and $n$ the number of agents. For a constant number of agents with submodular valuations, Garg et al.~\cite{GargKK20} provide an $e/(e-1)$-approximation algorithm and show that, even in this setting, improving upon $e/(e-1)$ is {\rm NP}-hard.

%Theorem 4 in \cite{NR14} provides the stated approximation guarantee of $(m ? n + 1)$. The proof of the theorem can in fact be easily modified to obtain an approximation ratio of $(m/n)$.

In the context of allocating indivisible goods, two other well-studied welfare objectives are the social welfare and the egalitarian welfare. These represent, respectively, for an allocation, the average valuation of the agents and the minimum valuation of any agent. For the social welfare objective, %under additive valuations, the optimal algorithm simply assigns each good to the agent that values it the most. Beyond additive valuations, 
a tight approximation factor of $e/(e-1)$ is known %for maximizing social welfare 
under submodular valuations \cite{vondrak2008optimal}. For subadditive valuations, Feige~\cite{Feige09} shows that social welfare maximization admits a polynomial-time $2$-approximation, assuming oracle access to \emph{demand queries}.\footnote{A demand-query oracle, when queried with prices $p_1, \ldots, p_m \in \mathbb{R}$ associated with the $m$  goods, returns $\max_{S \subseteq [m]} \left( v(S) - \sum_{j \in S} p_j \right)$, for an underlying valuation function $v$. The current paper works with more basic value oracle, which when queried with a subset of goods returns the value this subset. Any value query can be simulated via a polynomial number of demand queries. However, the converse is not true~\cite{nisan2007algorithmic}.} 

For maximizing egalitarian welfare under additive valuations, Chakrabarty et al.~\cite{ChakrabartyCK09} provide an $\widetilde{O}(n^\varepsilon)$-approximation algorithm that runs in time $n^{O(1/\varepsilon)}$, for any $\varepsilon > 0$. Under submodular valuations, egalitarian welfare maximization admits an $\widetilde{O}(n^{1/4}m^{1/2})$-approximation algorithm~\cite{GoemansHIM09}. Khot and Ponnuswami~\cite{Khot2007ApproximationAF} provide a $2n$-approximation algorithm for maximizing egalitarian welfare under subadditive valuations. As a lower bound, with submodular valuations, egalitarian welfare cannot be approximated within a factor of $2$, unless ${\rm P} = {\rm NP}$~\cite{BezakovaD05}. %Indeed, for egalitarian welfare, there is a notable gap between the best known approximation ratio and the hardness result.

In this work we develop a unified treatment of fairness and efficiency objectives, including the welfare functions mentioned above. In particular, we develop an approximation algorithm for computing allocations that maximize the \emph{generalized mean} of the agents' valuations.  Formally, for exponent parameter $p \in \mathbb{R}$, the $p^{\text{th}}$ \emph{generalized mean} of a set of $n$ positive reals $v_1, v_2, \ldots, v_n$ is defined as $\left( \frac{1}{n} \sum_{i=1}^n v_i^p \right)^{1/p}$. For an allocation (partition) $\mathcal{A} = (A_1, \ldots, A_n)$ of the indivisible goods among the $n$ agents, we define the \emph{$p$-mean welfare} of $\mathcal{A}$ as the generalized mean of the values $(v_i(A_i))_{i \in [n]}$; here $v_i(A_i)$ is the value that agent $i$ has for the bundle $A_i$ assigned to it. Indeed, with different values of $p$, the $p$-mean welfare encompasses a range of objectives: it corresponds to the social welfare (arithmetic mean) for $p=1$, the Nash social welfare (geometric mean) for $p \to 0$, and the egalitarian welfare for $p \to -\infty$. In fact, $p$-mean welfare functions with $p \in (-\infty, 1]$ exactly correspond to the collection of functions characterized by a set of natural axioms, including the Pigou-Dalton transfer principle~\cite{Moulin03}. Hence, $p$-mean welfare functions, with $p \in (-\infty, 1]$, constitute an important and axiomatically-supported family of objectives.  \\
%Hence,  and are thus well-studied as important welfare objectives~\cite{Moulin03}. 

%Approximations for each of these objectives have been studied as well. For the utilitarian welfare and additive valuations, the optimal algorithm simply assigns each good to the agent that values it the most. Beyond additive utilities, tight approximation factors of $2$ and $e/(e-1)$ are known for subadditive and XOS utility functions respectively~\cite{Feige09}. For the egalitarian welfare, an algorithm that returns an $\widetilde{O}(n^\epsilon)$ approximation in time $n^{O(1/\epsilon)}$ for any $\epsilon > 0$ is known for additive utilities~\cite{ChakrabartyCK09}, and an $\widetilde{O}(n^{1/4}m^{1/2})$-approximation algorithm for submodular utilities~\cite{GoemansHIM09}.  The best known lower bound for approximation is 2 for additive utilities~\cite{BezakovaD05}. 

%If all agents have identical subadditive valuations, it is known how to obtain in polynomial time an allocation that is a constant factor approximation to the generalized mean welfare for all $p \le 1$ simultaneously (ref?) (also check lower bounds).

%\paragraph*{Our Contributions} 
\noindent
{\bf Our Contributions.} We develop a polynomial-time algorithm that, given a fair division instance with subadditive valuations and parameter $p \in (-\infty, 1]$, finds an allocation with $p$-mean welfare at least $\nicefrac{1}{8n}$ times the optimal $p$-mean welfare (Theorem \ref{theorem:approximation-guarantee-for-p}). Our algorithm uses the standard value oracle model which, when queried with any subset of goods and an agent $i$, returns the value that $i$ has for the subset. For different values of $p$, our algorithm changes minimally, differing only in the weights of edges for a computed matching. We thus present a unified analysis for this broad class of welfare functions, suggesting further connections between these objectives than the previously mentioned axiomatization. Our result matches the best known $O(n)$-approximation for egalitarian welfare~\cite{Khot2007ApproximationAF} and improves upon the $O(n \log n)$-approximation guarantee of Garg et al.~\cite{GargKK20} for Nash social welfare with submodular valuations. Arguably, our algorithm (and the analysis) is simpler than the one developed in \cite{GargKK20} and simultaneously more robust, since it obtains an improved approximation ratio for subadditive valuations and a notably broader class of welfare objectives. 
%requires access to the underlying valuations through a basic value oracle

%matching computed in the first step.

For clarity of exposition, we first present an $8n$-approximation algorithm for maximizing Nash social welfare under subadditive valuations (Theorem \ref{theorem:approximation-guarantee}). We then generalize the algorithm to the class of $p$-mean welfare objectives.

%%%For egalitarian welfare, our algorithm is the first to achieve an approximation ratio independent of the number of goods $m$, for both submodular and subadditive valuations. Furthermore, note that in instances with number of goods $m \geq n \sqrt{n}$, our result instantiated for egalitarian welfare outperforms the one obtained in~\cite{GoemansHIM09}.

%We note that for the Nash social welfare, our algorithm improves upon the $O(n \log n)$ by Garg et al. for submodular valuations~\cite{GargKK20} and is the first approximation independent of the number of goods $m$ for subadditive valuations. 

%The most basic oracle considered in literature answers \emph{value queries}: given a subset of the indivisible goods, the value oracle returns the value this subset.Our algorithm uses the oracle model for the utility functions of the agents with value queries, that given a set of goods $S \subseteq \goods$ and an agent $i$, returns the value $v_i(S)$. 

We complement these algorithmic results by adapting a result of Dobzinski et al.~\cite{DobzinskiNS10} to show that for XOS valuations, any $O(n^{1-\varepsilon})$-approximation for $p$-mean welfare requires an exponential number of value queries (Section \ref{section:lower-bound}). Hence, in the value oracle model, our approximation guarantee is essentially tight for XOS and, hence, for subadditive valuations. We note that these are the first polynomial lower bounds on approximating either the Nash social welfare or the egalitarian welfare.% and are obtained by adapting the result of Dobzinski et al.~\cite{DobzinskiNS10}.

Nguyen and Rothe~\cite{NR14} obtain an $(m-n+1)$-approximation guarantee for maximizing Nash social welfare with subadditive valuations. We establish two extensions of this result. First, we show that, under subadditive valuations, an $(m-n+1)$-approximation for the $p$-mean welfare can be obtained for all $p \le 0$. However, for $0< p < 1$, we establish that it is {\rm NP}-hard to obtain an $(m-n+1)$-approximation, even under additive valuations. An analogous hardness result holds for $p=1$ with submodular valuations.\\

%We observe that for the Nash social welfare, our algorithm is tight in another sense as well. Combined with the $m/n$ approximation of The lower bound of Dobzinski et al., this gives an approximation ratio of $\min \{ O(n),m/n\}$.  actually constructs an instance with $n$ agents and $n^2$ goods, and hence shows 
%{\bf Outline:}  

\noindent
{\bf Independent Work.} In work independent of ours, Chaudhury et al.~\cite{chaudhury2020fair} also obtain an $O(n)$-approximation algorithm for maximizing generalized $p$-means under subadditive valuations. Their approach varies significantly from the current paper and, in particular, builds upon results on finding allocations that are approximately envy-free up to any good ($\EFX$). Notably their algorithm computes allocations that satisfy additional fairness properties, including $\EFone$ and either of two approximate versions of $\EFX$. \\

Section \ref{section:nash} presents our approximation algorithm for maximizing Nash social welfare. Then, Section \ref{section:p-mean} shows that we can extend the algorithm for Nash social welfare to obtain the stated approximation bound for $p$-mean welfare. The tightness of these results is established in Section \ref{section:lower-bound}. Section~\ref{section:m-n} presents the results for the $(m-n+1)$-approximation guarantees.

\section{Notation and Preliminaries}
\label{section:notation}

An instance of a fair division problem is a tuple $\langle [m], [n],  \{v_i \}_{i=1}^n \rangle$, where $[m]= \left\{1,2,\ldots, m \right\}$ denotes the set of $m \in \mathbb{N}$ indivisible {goods} that have to be allocated (partitioned) among the set of $n \in \mathbb{N}$ agents, $[n]=\{1, 2, \ldots, n\}$. Here, $v_i: 2^{[m]} \mapsto \mathbb{R}_+$ represents the valuation function of agent $i \in [n]$. Specifically, $v_i(S) \in \mathbb{R}_+$ is the value that agent $i$ has for a subset of goods $S \subseteq [m]$. For $g \in [m]$ and $i\in [n]$, write $v_i(g)$ to denote agent $i$'s value for the good $g$, i.e., it denotes $v_i(\{g\})$. 

We will assume throughout that the valuation function $v_i$ for each agent $i\in [n]$ is (i) nonnegative: $v_i(S) \geq 0$ for all $S \subseteq [m]$, (ii) normalized: $v_i(\emptyset) = 0$, (iii) monotone: $v_i(A) \leq v_i(B)$ for all $A \subseteq B \subseteq [m]$, and (iv) {subadditive}: $v_i(A \cup B) \leq v_i(A) + v_i(B)$ for all subsets $A, B \subseteq [m]$. 

Submodular and XOS (fractionally subadditive) valuations constitute subclasses of subadditive valuations. Formally, a set function $v: 2^{[m]} \mapsto \mathbb{R}_+$ is said to be submodular if it satisfies the \emph{diminishing marginal returns} property: $v(A \cup \{g\}) - v(A) \geq v(B \cup \{g\}) - v(B)$, for all subsets $A \subseteq B \subset [m]$ and $g \in [m] \setminus B$.  A set function, $v: 2^{[m]} \mapsto \mathbb{R}_+$, is said to be XOS if it is obtained by evaluating the maximum over a collection of additive functions $\{f_r \}_{r \in [L]}$, i.e., $v(S) \coloneqq \max_{1 \leq j \leq L} \left\{ f_r(S) \right\}$, for each subset $S \subseteq [m]$.\footnote{Here, $L$ can be exponentially large in $m$.}

We use $\Pi_n([m])$ to denote the collection of all $n$ partitions of the indivisible goods $[m]$. An \emph{allocation} is an $n$-partition $\mathcal{A} = (A_1,\ldots, A_n) \in \Pi_n([m])$ of the $m$ goods. Here, $A_i$ denotes the subset of goods allocated to agent $i \in [n]$ and will be referred to as a \emph{bundle}. 

%along with agent-specific weights, $\eta_i \geq 0$ for each $i \in [n]$, 
%weighted
Given a fair division instance $\mathcal{I}=\langle [m],[n], \{v_i \}_i \rangle$, the \emph{Nash social welfare} of allocation $\mathcal{A}$ is defined as the geometric mean of the agents' valuations under $\mathcal{A}$: $ \ \NSW (\mathcal{A}) \coloneqq \left( \prod_{i=1}^n   v_i(A_i)  \right)^\frac{1}{n}$.
%\begin{align*}
%\NSW (\mathcal{A}) \coloneqq \left( \prod_{i=1}^n   v_i(A_i)  \right)^\frac{1}{n}
%\end{align*}
%Here, $\eta _i \geq 0$ is the weight associated with agent $i$. 
%The standard (unweighted) Nash social welfare is obtained by setting $\eta_i = 1$ for all $i \in [n]$.

We will throughout use $\mathcal{N}^* = (N^*_1,\ldots, N^*_n)$ to denote an allocation that maximizes the Nash social welfare for a given fair division instance. % $\mathcal{I}= \langle [m], [n],  \{v_i \}_{i=1}^n \rangle$, i.e., 
%\begin{align*}
%$\NSW (\mathcal{N}^*) \geq \NSW (\mathcal{A})$ for all allocations $\mathcal{A} \in \Pi_n([m]).$
%\end{align*}
We refer to $\mathcal{N}^*$ as a Nash optimal allocation. An allocation $\mathcal{P}=(P_1,\dots ,P_n)$ is an $\alpha$-approximate solution (with $ \alpha \geq 1$) of the Nash social welfare maximization problem if $\NSW (\mathcal{P}) \geq \frac{1}{\alpha} \NSW (\mathcal{N}^*)$. 

Besides the Nash social welfare, we address a family of objectives defined by considering the \emph{generalized means} of agents' valuations. In particular, for parameter $p \in \mathbb{R}$, the the $p^\text{th}$ generalized (H\"{o}lder) mean ${\rm M}_p(\cdot)$ of $n$ nonnegative numbers $x_1,\ldots , x_n \in \mathbb{R}_+$ is defined as ${\rm M}_p \left(x_1, \ldots, x_n \right) \coloneqq \left(  \frac{1}{n} \sum \limits _{i=1}^n  x_i^p \right )^\frac{1}{p}$. 

Parameterized by $p$, this family of functions captures multiple fairness and efficiency measures. In particular, when $p=1$, ${\rm M}_p$ reduces to the arithmetic mean. In the limit, ${\rm M}_p$ is equal to the geometric mean as $p$ tends to zero. In addition, $\lim_{p \to -\infty} {\rm M}_p \left(x_1, \ldots, x_n \right)  = \min\{x_1, x_2, \ldots, x_n\}$.

%Given that generalized means constitute a parameterized family of welfare functions, 
We define the \emph{$p$-mean welfare}, ${\rm M}_p(\mathcal{A})$, of an allocation $\mathcal{A}=(A_1, A_2, \ldots, A_n)$ as  
\begin{align*}
{\rm M}_p(\mathcal{A}) & \coloneqq  {\rm M}_p\left( v_1(A_1), \ldots, v_n(A_n) \right) = \left(  \frac{1}{n} \sum_{i=1}^n v_i(A_i)^p \right)^{1/p} \, . %\label{eq:gen-mean}
\end{align*}

With $p$ equal to one, zero, and $-\infty$, the $p$-mean welfare corresponds to the (average) social welfare, Nash social welfare, and egalitarian welfare,  respectively.

The following proposition implies that for any $p \le - n \log n$, if instead of the $p$-mean welfare, we maximize the egalitarian welfare, then the resulting allocation loses a negligible factor in the approximation ratio. The proof of this proposition is deferred to Appendix \ref{app:notation}. 
 
\begin{restatable}{proposition}{PropMp}  
For any $n$ nonnegative numbers $x_1, \ldots, x_n \in \mathbb{R}_+$ and $p \le - n \log n$, we have
\begin{align*}
\M_{-\infty}(x_1, \ldots, x_n) ~\le ~ \M_{p}(x_1, \ldots, x_n) ~\le ~2^{1/n} \ \M_{-\infty}(x_1, \ldots, x_n) \, .
\end{align*}
\label{prop:large-p}
\end{restatable}

\section{An $8n$-Approximation for Nash Social Welfare}
\label{section:nash}
%\subsection{Algorithm and Main Result} 
\label{section:nsw}
This section presents an efficient $8n$-approximation algorithm for the Nash social welfare maximization problem, under subadditive valuations. Our algorithm, Algorithm \ref{Alg} ($\textsc{Alg}$), requires access to the valuation functions through basic value queries, i.e., it only requires an oracle which, when queried with a subset of goods $S \subseteq [m]$ and an agent $i \in [n]$, returns $v_i(S) \in \mathbb{R}_+$.

\floatname{algorithm}{Algorithm}
\begin{algorithm}[ht]
	\caption{\textsc{Alg}} \label{Alg}%Computation of approximate $p$-mean maximizing allocation  
	\textbf{Input:} Instance $\mathcal{I}= \langle [m],[n], \{v_i\}_{i=1}^n \rangle$ with value oracle access to the valuation functions $v_i$s. \\ \textbf{Output:} An allocation $\mathcal{B} = (B_1, B_2, \ldots, B_n)$ 
	\begin{algorithmic}[1]
		%  \STATE Initialize  bundle $B^0_i= \emptyset $ for all $i\in [n]$ %Initialize the set of agents ${A}=[n]$, the set of goods $G=[m]$, and bundle $B^0_i= \emptyset $ for all $i\in [n]$
		\STATE Initialize iteration count $t=0$ and define ${\rm SAT}_t = \emptyset$ and ${\rm UNSAT}_t = [n]$%and $\gamma^1_i=  v_i([m])$ for all $i \in [n]$ 
		\FOR {$i \in [n]$}
		\STATE Sort the goods in $[m]$ in descending order of value such that $v_i(g_1)\geq \dots \geq v_i(g_m)$
		\IF{$v_i \left( [m] \setminus \{g_1,\ldots , g_{2n}\} \right) =0$}
		\STATE Set $\gamma ^t_i =0$
		%\ELSIF{$v_i \left([m]\setminus \{g_1,\ldots , g_{2n}\}\right) > 0$}
		%\STATE $\gamma ^t_i = v_i([m])$
		\ELSE \STATE $\gamma ^t_i = v_i([m])$
		\ENDIF 
		\ENDFOR
		%        \STATE  Define the sets {\rm SAT}$_1 \coloneqq \{ i \in [n] \mid v_i(B_i^0) \geq \gamma _i^0 \}$ and {\rm UNSAT}$_1 = \{ i \in [n] \mid  v_i(B_i^0) < \gamma _i^0\}$
		% \COMMENT{The initializations ensure that ${\rm UNSAT}_1 = [n]$.}
		\WHILE  {{\rm UNSAT}$_t \neq \emptyset$} \label{SATloop}
		\STATE Consider the bipartite graph $ \left( [n] \cup [m],  [n]  \times [m], \{ w(i,g) \}_{i \in [n], g \in [m]} \right)$ with weight of edge $(i,g) \in [n] \times [m] $ set as $w(i,g) =  \log \left(v_i(g) + \gamma _i^t\right)$  \label{step:matching}
		\STATE Compute a left-perfect maximum-weight matching, $\pi^t$, in this bipartite graph \label{step:compute-matching}
		\STATE Set $ G =[m] \setminus \{ \pi ^t (i)\}_{i\in[n]}$ and ${A} =[n]$ \label{step:original-G}
		\WHILE{there exists ${a}' \in A$ and ${g'} \in G$ such that $v_{a'} (g') \geq \frac{1}{2n} v_{a'} (G)$} \label{high_value_loop}
		\STATE Set $B^t_{a'} = \{ g' \}$ and update $ G \leftarrow G \setminus \{ g' \}$ along with $A \leftarrow A \setminus \{ a' \}$ \label{step:singleton}
		\ENDWHILE
		%   \STATE Index the remaining set of agents such that ${A} = \{ 1,2,\dots, |A|\}$. 
		\STATE Set $(B^t_i)_{i \in A} =  \textsc{MovingKnife} \left(G, A, \{v_i\}_{i \in A} \right)$ \label{step:moving-knife}
		\STATE Define {\rm SAT}$_{t+1} = \left\{i \in [n] \mid v_i(B_i^t) \geq  \gamma _i^t \right\}$ and set  $\gamma _i^{t+1} = \gamma _i^{t}$ for each $i\in$ {\rm SAT}$_{t+1}$
		\STATE Define {\rm UNSAT}$_{t+1}= \left\{i \in [n] \mid  v_i(B_i^t) < \gamma _i^t \right\}$ and set $\gamma _i^{t+1} = \left ( 1- \frac{1}{m}\right ) \gamma _i^{t}$ for each $i \in$ {\rm UNSAT}$_{t+1}$  
		\STATE Update $t\leftarrow t+1$
		\ENDWHILE
		\RETURN allocation $\left( B^{t-1}_1 \cup \{\pi^{t-1}(1) \}, B^{t-1}_2 \cup \{\pi^{t-1}(2) \}, \ldots, B^{t-1}_n \cup \{\pi ^{t-1} (n)\} \right)$
	\end{algorithmic}
\end{algorithm}

%, and briefly compare it with the work of Garg et al.~\cite{GargKK20}.
We first describe the ideas behind our algorithm. Write $\mathcal{N}^* = \{N_1^*, \ldots, N_n^*\}$ denote a Nash optimal allocation in the given instance and let us, for now, assume that the agents have additive valuations, i.e., for all agents $i \in [n] $ and subset of goods $S \subseteq [m]$, we have $v_i(S) = \sum_{g \in S} v_i(g)$. In the following two cases, we can readily obtain an $O(n)$ approximation. In the first case, each agent has a few  ``high-value'' goods, i.e., each agent $i$ has a good ${g}'_i \in N_i^*$ with the property that $v_i({g}'_i) \geq v_i(N_i^*)/n$. In such a setting, we can construct a complete bipartite graph with agents $[n]$ on one side and all the goods $[m]$ on the other. Here, the weight of edge $(i,g) \in [n] \times [m]$ is set to be $\log \left(v_i(g) \right)$. In this bipartite graph, the matching $(i, {g}'_i)_{i \in [n]}$ has Nash social welfare at least $\nicefrac{1}{n}$ times the optimal and, hence, this also holds for a left-perfect maximum-weight matching in this graph. 

In the second case, all goods are of ``low-value'', i.e., for all $i \in [n]$ and $g \in [m]$ we have $v_i(g) \leq v_i(N_i^*)/(2n)$. Here again an $O(n)$ approximation can be obtained via a simple round-robin algorithm, wherein the agents (in an arbitrary order) repeatedly pick their highest valued good from those remaining. At a high level, our algorithm stitches together these two extreme cases by first matching high-value goods and then allocating the low-value ones.

%Our algorithm, as well as the SMatch (for additive valuations) and RepReMatch (for submodular valuations) algorithms of Garg et al.~\cite{GargKK20}, stitch together these two extreme cases. Hence, at a high level, here all of the algorithms consist of a matching step for high-value goods, followed by steps to allocate the low-value goods. %In the case  of Garg et al., the low-value goods are also handled by matching algorithms. In our case, for subadditive valuations, we use a different algorithm.

%both the SMatch (for additive valuations) and RepReMatch (for submodular valuations) algorithms of
%The dividing line between high-value and low-value goods is the following quantity: For each agent $i$, define
We connect the two cases by considering the following quantity for each agent $i \in [n]$
\begin{align}\label{eq:prop-lb}
\ell_i \coloneqq \min_{S \subseteq [m]: |S| \leq 2n} \  \frac{1}{2n} \ v_i \left([m] \setminus S \right) \, .
\end{align}

%\noindent %That is, $\ell_i$ is the minimum value agent $i$ has for the remaining set of goods upon removal of $2n$ goods.

That is, $\ell_i$ is the (near) proportional value that each agent is guaranteed to achieve even after the removal of any $2n$-size subset of goods. Our algorithm leverages the following existential guarantee (Lemma \ref{lemma:upper-bound}): there necessarily exists a good $\widehat{g}_i \in N_i^*$ with the property that 
%Our algorithm, and roughly the ones developed by Garg et al.~\cite{GargKK20}, depend on the following existential result: there necessarily exists a good $\widehat{g}_i \in N_i^*$ with the property that 
\begin{align}
v_i(\widehat{g}_i) + \ell_i & \geq \frac{1}{4n}v_i(N^*_i) \, .
\label{eq:keyineq1}
\end{align}
This result ensures that, a single high-value good (in particular, $\widehat{g}_i$) coupled with a $2n$-approximation to all the low-value goods (i.e., $\ell_i$), is sufficient to ensure a $4n$-approximation \emph{for each agent}. At this point, if we could (i) explicitly compute $\ell_i$ for each agent $i$ and (ii) for any size-$n$ subset of goods $S$, assign the remaining goods $[m] \setminus S$ such that each agent gets a bundle of value at least $\ell_i$, then we would be done. This follows from the observation that in the complete bipartite graph $([n] \cup [m], [n] \times [m])$ with weight of edge $(i,g)$ set to $\log \left( v_i(g) + \ell_i \right)$, the weight of the matching $(i, \widehat{g}_i)_i$ is a $4n$ approximation to the optimal Nash social welfare by equation~\eqref{eq:keyineq1} and, hence, the same guarantee holds for a maximum-weight matching in the graph. Condition (ii) ensures that each agent also receives at least $\ell_i$ after the initial assignment of the $n$ matched goods.   

For additive valuations, both conditions (i) and (ii) can be satisfied. This template was employed in the SMatch algorithm (for additive valuations) of Garg et al.~\cite{GargKK20}. However, for submodular (and subadditive) valuations, the quantity $\ell_i$ is hard to approximate within a sub-linear factor~\cite{svitkina2011submodular}. 

%Here, our work diverges considerably from~\cite{GargKK20}. 
Therefore, instead of satisfying condition (i) explicitly, we maintain an upper bound $\gamma_i \ge \ell_i$ for each agent $i$. Our algorithm first obtains a maximum weight matching in the bipartite graph between agents and goods with the weight of edge $(i,g) \in [n] \times [m]$ set to $\log (v_i(g) + \gamma_i)$. It assigns all the matched goods to the respective agents, removes these goods from further consideration in this iteration, and then carries out a procedure (described below) to ensure condition (ii). If, for agent $i$, the bundle obtained in this procedure (i.e., the bundle obtained for $i$ after removing the matched goods) has value less than $\gamma_i$, then we multiplicatively reduce the (over) estimate $\gamma_i$ for $i$ and repeat the algorithm. 

The procedure towards satisfying condition (ii) consists of two steps. Let $G$ be the set of goods that remain once we remove the matched $n$ goods from $[m]$. In the first step, if there exists an agent $i$ and a good $g \in G$ such that $v_i(g) \ge v_i(G)/(2n)$, we assign $g$ to $i$ and remove both from further consideration. An agent thus removed has value $\ell_i$ from the assigned good; note that, by definition,  $\ell_i \leq v_i(G)/(2n)$. After this step, we observe that $v_i(g) \leq v_i(G)/(2n)$ for each remaining agent $i$ and good $g$. In the second step, we run a \emph{moving knife} subroutine (Algorithm \ref{MovingKnife}) on the goods that are still unassigned. In this subroutine, the goods are initially ordered in an arbitrary fashion.  A hypothetical knife is then moved across the goods from one side until an agent $i$ (who has yet to receive a bundle) calls out that the goods covered so far have a collective value of at least $v_i(G)/(2n)$ for her. These covered goods are then allocated to said agent $i$ and both the agent as well as this bundle is removed from further consideration. We show that this allocation satisfies condition (ii), i.e., the bundle assigned to each agent in this procedure has value at least $\ell_i$ (but it may be lower than the overestimate $\gamma_i$).

Since we can guarantee $\ell_i$ for each agent $i$, irrespective of which goods are removed in the matching step, $\gamma_i$ never goes below $\ell_i$, for any agent. Hence, at some point, every agent $i$ receives a bundle of value at least $\gamma_i$ in the above two steps. We show that these bundles, with the goods matched with each agent, provide an $8n$ approximation to the optimal Nash social welfare.

It is relevant to note we use $\ell_i$ solely for the purposes of analysis. Our algorithm executes with the overestimate $\gamma_i$ and keeps reducing this value till it is realized (in the two-step procedure) for all the agents.  
 
As mentioned previously, the SMatch algorithm (developed for additive valuations) of Garg et al.~\cite{GargKK20} relies of conditions (i) and (ii). However, for submodular valuations their work diverges considerably from the current approach. In particular, the RepReMatch algorithm (developed for submodular valuations) in \cite{GargKK20} first finds a set of goods ${\goods}$ with the property that in the bipartite graph between all the agents and ${\goods}$, there is a matching wherein every agent is matched to a good with value at least as much as her highest valued good in $N_i^*$. To ensure this property the cardinality of ${\goods}$ needs to be $n \log n$. Intuitively, this requirement leads to a lower bound of $\Omega(n \log n)$ on the approximation ratio obtained in \cite{GargKK20}. Furthermore, the steps in their algorithm to ensure condition (ii) do not extend to subadditive valuations either. Specifically, Garg et al. \cite{GargKK19} note that their algorithm gives an approximation ratio of $\Omega(m)$ for the case of subadditive valuations. The $2n$-approximation of Khot and Ponnuswami for egalitarian welfare~\cite{Khot2007ApproximationAF} first guesses the optimal egalitarian welfare $b$, and uses this to partition the goods into ``large'' ones (those with value higher than $b/n$) and ``small'' ones, for each agent. It then tries to ensure every agent receives a bundle with valuation at least $b/n$. For Nash social welfare, guessing just a single value does not appear to help, since the Nash social welfare depends on the valuation of each agent. \\

The following theorem constitutes our main result for Nash social welfare. 
 
%\newtheorem{thm}{Theorem}[section]
%\newcommand{\thistheoremname}{Main Theorem }

%\newtheorem*{genericthm*}{\thistheoremname}
%\newenvironment{namedthm*}[1]
%{\renewcommand{\thistheoremname}{#1}%
%	\begin{genericthm*}}
%	{\end{genericthm*}}

\begin{theorem} \label{theorem:approximation-guarantee}	
Let $\mathcal{I} = \langle [m], [n], \{v_i\}_{i=1}^n \rangle$ be a fair division instance in which the valuation function $v_i$, of each agent $i \in [n]$, is nonnegative, monotone, and subadditive. Given value oracle access to $v_i$s, the algorithm \textsc{Alg} computes an $8 n$ approximation to the Nash optimal allocation in polynomial time.
%in polynomial-time an allocation $\mathcal{A}$ that provides an $O(n)$-approximation to the optimal nash social welfare, i.e., $NSW(\mathcal{A})\geq \left ( 1- \frac{1}{m}\right )\frac{1}{8n} NSW(\mathcal{N})$, where $NSW(.)$ denotes the Nash Social Welfare with respect to a given allocation and $\mathcal{N}$ denotes the Nash optimal allocation for the given instance. 
\end{theorem}

\floatname{algorithm}{Algorithm}
%\begin{minipage}[t][5.8cm][b]{0.41\textwidth}
\begin{algorithm}[ht]
  \caption{\textsc{MovingKnife} } \label{MovingKnife}
   \textbf{Input:} Instance $\mathcal{J} = \langle G, {A}, \{v_i \}_{i \in A} \rangle$ with value oracle access to the valuation functions $v_i$s \\
    \textbf{Output:} An allocation $\mathcal{P} = (P_1, P_2, \ldots, P_{|{A}|})$ 
  \begin{algorithmic} [1]
        \STATE Initialize $S=\emptyset$,  $\widehat{G}=G$, $\widehat{A} = A$, and bundle $P_i =\emptyset$ for all  $i \in {A}$. 
        \WHILE {$\widehat{G} \neq \emptyset $ and $ \widehat{A} \neq \emptyset $}
        \STATE Select any arbitrary good ${g} \in \widehat{G}$ and update $S\leftarrow S \cup \{{g} \}$ along with $\widehat{G}\leftarrow \widehat{G} \setminus \{ {g} \}$.
        \IF{for some agent $\widehat{a} \in \widehat{A}$ we have $v_{ \widehat{a} } (S) \geq \frac{1}{2n} v_{ \widehat{a}} (G)$ } \label{step:cut-condition}
        \STATE Set $P_{ \widehat{a} } = S$ and update $\widehat{A} \leftarrow \widehat{A} \setminus \{ \widehat{a} \}$ along with $S \leftarrow \emptyset$. 
        \ENDIF
        \ENDWHILE
        \IF{$\widehat{G} \neq \emptyset$}
        \STATE $P_{|A|} \leftarrow P_{|A|} \cup \widehat{G}$
        \ENDIF
        
        \RETURN allocation $\mathcal{P} = (P_1,\ldots ,P_{|{A}|}).$
     
  \end{algorithmic}
\end{algorithm}

%\subsection{Supporting Lemmas} \label{section:lemmas}
\begin{comment}
To prove Theorem \ref{theorem:approximation-guarantee}, we first establish some supporting results. Specifically, for analysis, we will consider the (near) proportional value that each agent can achieve even after the removal of any $2n$-size subset of goods. Specifically, in instance $\mathcal{I}=\langle [m],[n], \{v_i\}_i \rangle$, for each agent $i \in [n]$ define 
\begin{align}\label{eq:prop-lb}
\ell_i \coloneqq \min_{S \subset [m]: |S| \leq 2n} \  \frac{1}{2n} \ v_i \left([m] \setminus S \right)
\end{align}

It is relevant to note that $\ell_i$ is defined for the purposes of analysis. As was pointed out in~\cite{GargKK20}, in and of itself, the quantity $\ell_i$ is hard to approximate within a sub-linear factor.\footnote{This hardness result holds even under submodular valuations.} In particular, our algorithmic result only requires a realizable overestimate of $\ell_i$ (see, e.g., Lemma~\ref{lemma:lb-bundles}) and is not affected by the fact that $\ell_i$ is hard to approximate. 
\end{comment}

The following lemma proves inequality~\eqref{eq:keyineq1}. We state and prove it for an arbitrary allocation $\mathcal{A}^*=(A^*_1, \ldots, A^*_n)$, rather than just for the Nash optimal allocation. % ides a useful upper bound on the value that each agent $i$ achieves under any allocation $\mathcal{A}^*=(A^*_1, \ldots, A^*_n)$ in terms of the most valued good in each bundle $A^*_i$ and $\ell_i$.  
%We now state the upper bounds on values $v_i(N^*_i)$s in terms of the above-mentioned quantities. 

\begin{lemma} \label{lemma:upper-bound}
Let $\mathcal{I}=\langle [m], [n], \{v_i\}_{i=1}^n  \rangle$ be a fair division instance with monotone, subadditive valuations and let $\mathcal{A}^*=(A^*_1,\ldots, A^*_n)$ be any allocation in $\mathcal{I}$. Let $\widehat{g}_i$ be the most valued (by $i$) good in $A^*_i$ (i.e., $\widehat{g}_i \coloneqq \argmax_{g\in A^*_i} v_i(g)$) and $\ell_i$ be as defined in~\eqref{eq:prop-lb}. Then, for each agent $i \in [n]$ 
\begin{align*}
v_i(\widehat{g}_i) + \ell_i & \geq \frac{1}{4n}v_i(A^*_i).
\end{align*}
\end{lemma}

\begin{proof}
Consider any agent $i\in [n]$  and note that $\ell_i \geq 0$. We will establish the lemma by considering two complementary cases.

\noindent
{Case {\rm I}}: There exists a good $g_i \in A^*_i$ with the property that $v_i(g_i) \geq \frac{1}{4n}v_i(A^*_i)$. Since $\widehat{g}_i$ is the most valued good in $A^*_i$, we have $v_i(\widehat{g}_i) \geq v_i(g_i)$ and the desired inequality follows.\\ %$v_i(\widehat{g}_i) + \ell_i \geq v_i(g_i) + \ell_i \geq \frac{1}{4n} v_i(N^*_i) $. 

\noindent
{Case {\rm II}:} For all goods $g \in A^*_i$, $v_i(g) < \frac{1}{4n}v_i(A^*_i)$. Recall that $\ell_i \coloneqq \min \limits _{S \subseteq [m], |S| \leq 2n} \frac{1}{2n} v_i([m] \setminus S)$. Let $S^*$ be the set $S$ that induces $\ell_i$, i.e., $\ell_i = \frac{1}{ 2n} v_i([m] \setminus S^*)$. Monotonicity of $v_i$ ensures that $|S^*| = 2n$ and 
\begin{align}
\ell_i = \frac{1}{ 2n} v_i([m] \setminus S^*) \geq \frac{1}{2n} v_i(A^*_i \setminus S^*) \, . \label{ineq:l-i-upper}
\end{align}

Furthermore, given that in the current case $v_i(g) <\frac{1}{4n} v_i(A^*_i)$ for all $g \in A^*_i$, we have 
\begin{align}
v_i(A^*_i \cap S^*) & \leq \sum_{g \in A^*_i \cap S^*} v_i(g) <  \sum_{g \in A^*_i \cap S^*} \frac{1}{4n} v_i(A^*_i) \leq \frac{|S^*|}{4n} v_i(A^*_i) = \frac{1}{2} v_i(A^*_i) \, . \label{ineq:cap-value}
\end{align}
Here, the first inequality follows from the fact that $v_i$ is subadditive and the last since $|S^*| = 2n$.  

Therefore, we obtain the desired bound in terms of $\ell_i$:
%Inequalities (\ref{ineq:l-i-upper}) and (\ref{ineq:cap-value}) lead to 
\begin{align*}
\ell_i & \geq \frac{1}{2n} v_i(A^*_i \setminus S^*) \tag{via inequality (\ref{ineq:l-i-upper})} \\
& \geq \frac{1}{2n} \left( v_i(A^*_i) - v_i(A^*_i \cap S^*)  \right) \tag{$v_i$ is subadditive} \\
& \geq \frac{1}{4n}v_i(A^*_i) \tag{via inequality (\ref{ineq:cap-value})}
\end{align*}
Thus, the the stated inequality $v_i(\widehat{g}_i) + \ell_i \geq \frac{1}{4n}v_i(A^*_i)$ holds even in this case.
\end{proof}

The next lemma establishes the key property of Algorithm \ref{MovingKnife} (\textsc{MovingKnife}): if all the goods have low value for every agent, then \textsc{MovingKnife} returns a near-proportional allocation.

\begin{lemma}\label{lemma:movingknife}
Consider a fair division instance $ \langle G, A, \{v_i \}_{i \in A} \rangle$ wherein the agents have monotone, subadditive valuations. In addition, suppose for each agent $i \in A$ and good $g \in G$ we have $v_i(g)< \frac{1}{2n}v_i(G)$, where $n \geq |A|$. Then the allocation $(P_1,\ldots , P_{|{A}|})$ returned by Algorithm \ref{MovingKnife} (\textsc{MovingKnife}) satisfies $v_i(P_i)\geq \frac{1}{2n} v_i(G)$ for all $i \in A$. 
%\begin{align*}
%v_i(P_i)\geq \frac{1}{2n} v_i(G) \qquad \text{for all $i \in A$.}
%\end{align*}
\end{lemma}

\begin{proof}
Given instance $ \langle G, A, \{v_i \}_{i \in A} \rangle$, the \textsc{MovingKnife} algorithm (Algorithm \ref{MovingKnife}) considers the goods in an arbitrary order and adds these goods one by one into a bundle $S$ until an agent $\widehat{a}$ calls out  that its value for $S$ is at least $\frac{1}{2n}v_{\widehat{a}} (G)$. We assign these goods to agent $\widehat{a}$ and remove them---along with $\widehat{a}$---from consideration. The algorithm iterates over the remaining set of agents and goods. %Next, we will prove that one can repeat this process $|A|$ times to ensure that each agent $i \in A$ receives a bundle of value at least $\frac{1}{2n} v_i(G)$.  That is, 
We will show that the while loop in the \textsc{MovingKnife} algorithm terminates with $\widehat{A} = \emptyset$ and, hence, assigns to each agent a bundle of desired value.  

Consider an integer (count) $k \in \mathbb{N}$. Let $\widehat{G}$ and $\widehat{A}$ denote the set of goods and agents, respectively, that are left unassigned after $k$ agents are assigned bundles in \textsc{MovingKnife};
%assignments are performed in the while loop of \textsc{MovingKnife}; 
note that $|\widehat{A}| = |A| - k$. 
The arguments below establish that for each remaining agent $i \in \widehat{A}$, 
\begin{align}
v_i(\widehat{G}) & \geq \left( 1 - \frac{k}{n} \right) v_i(G) \label{ineq:prop-rolls}
\end{align}

Therefore, for any $k < |A| \leq n$, the set of unassigned goods $\widehat{G}$ is nonempty and even the last agent %that receives a bundle at the very end of the while loop 
(i.e., with $k = |A|-1$) receives a bundle of sufficiently high value. 

To prove~\eqref{ineq:prop-rolls}, consider any agent $i \in \widehat{A}$. Indeed, agent $i$ has not received any goods yet, but the $k$ agents in $A \setminus \widehat{A}$ have been assigned bundles. Let $S$ be a bundle assigned to some agent in $ A\setminus \widehat{A}$ (i.e., $P_j = S$ for some $j \in A \setminus \widehat{A}$) and $g'$ be the last good included in $S$. Step \ref{step:cut-condition} of the algorithm ensures that $v_i(S \setminus \{g'\}) < \frac{1}{2n} v_i(G)$; otherwise, $S \setminus \{g'\}$ would have been assigned to agent $i$.  Furthermore, the assumption (in the Lemma statement) gives us $v_i(g')  \leq \frac{1}{2n} v_i(G)$. Hence, using these inequalities and the subadditivity of $v_i$, we get $v_i(S) \leq v_i(S \setminus \{g'\}) + v_i(g') \leq \frac{1}{n} v_i(G)$. 

This inequality provides an upper bound on $v_i(G \setminus \widehat{G}) = v_i \left( \cup_{j \in A \setminus \widehat{A}} \ P_j \right) $, the total value of the set of goods assigned among the $k$ agents in $A\setminus \widehat{A}$. Specifically, by the subadditivity of $v_i$, $v_i(G \setminus \widehat{G}) \leq \frac{k}{n} v_i(G)$.  Therefore, $v_i(\widehat{G}) \geq v_i(G) - v_i(G \setminus \widehat{G})  \geq  \left( 1 - \frac{k}{n} \right) v_i(G)$.
%\begin{align*}
%\end{align*}

Overall, every agent $i \in A$ is eventually assigned a bundle of value at least $\frac{1}{2n} v_i(G)$ in the while loop.% and the stated claim follows. 
\end{proof}

% We now prove that this allocation, along with the goods assigned as singletons in the matching specified by algorithm \textsc{Alg} \ref{Alg}, is a good approximation to the optimal Nash Welfare of the given instance.

%Next, in Lemma~\ref{lemma:upper-bound} we develop a useful upper bound on the value that each agent achieves under a Nash optimal allocation $\mathcal{N}^*=(N^*_1, \ldots, N^*_n)$. 

%For each agent $i \in [n]$, write $\widehat{g}_i$ to denote the good of maximum value in $N^*_i$, i.e., $\widehat{g}_i \coloneqq \argmax_{g\in N^*_i} v_i(g)$. 

Next we show that in each iteration of the while loop in \textsc{Alg} (Algorithm \ref{Alg}), the value of the assigned bundle $B^t_i$ is at least as large as $\ell_i$. % (thereby establishing our algorithm satisfies condition (2)).

\begin{lemma}
	\label{lemma:lb-bundles}
	Given a fair division instance $\mathcal{I} = \langle [m], [n], \{v_i\}_{i=1}^n \rangle$ with subadditive valuations, let $B^t_i$ be the bundle assigned to agent $i \in [n]$ in the $t^{\text{th}}$ iteration (for $t \in \mathbb{N}$) of the outer while loop (Step \ref{SATloop}) in \textsc{Alg}. Then,  for all agents $i \in [n]$ and each iteration count $t$, we have $v_i(B_i^t)\geq  \ell _i$.
%	\begin{align*}
%	\end{align*}
\end{lemma}

\begin{proof}
During any iteration $t$ of the outer while loop (Step \ref{SATloop}) in \textsc{Alg} and for any agent $i \in [n]$, the bundle $B_i^t$ either consists of a single good of high value (Step \ref{step:singleton}), or of the set of goods assigned to agent $i$ obtained after executing the \textsc{MovingKnife} subroutine (Step \ref{step:moving-knife}). We will show that in both cases the stated inequality holds.

Recall that $ \ell_i \coloneqq \min \limits_{S \subseteq [m]: |S| \leq 2n} \  \frac{1}{2n} \ v_i \left([m] \setminus S \right)$. Equivalently, $\ell_i = \min \limits_{T \subseteq [m]: |T| \geq m - 2n} \ \frac{1}{2n} v_i(T)$. Therefore, we have
\begin{align}
\frac{1}{2n} v_i(T) \geq \ell_i \qquad \text{for any subset $T \subseteq [m]$ of size at least $(m-2n)$} \label{ineq:univ-l-i}
\end{align}

The relevant observation here is that, in any iteration $t$, the set of goods $G$ from which the bundles $B^t_i$s are populated satisfies $|G| \geq m- 2n$. Specifically, in the $t^{\text th}$ iteration, we start with $|G| = m -n$ (Step \ref{step:original-G}). Subsequently, the inner while loop (Step \ref{high_value_loop}) assigns at most $n$ goods and, hence, the number of goods passed on to the \textsc{MovingKnife} subroutine satisfies $|G| \geq m - 2n$.  

First, we note that the lemma holds for any agent $a'$ that receive a singleton bundle $B^t_{a'} = \{g'\}$ in Step \ref{step:singleton}: $v_{a'}(g') \geq \frac{1}{2n} v_{a'}(G) \geq \ell_{a'}$. Here, the first inequality follows from the selection criterion applied to $g'$ and the second inequality from equation (\ref{ineq:univ-l-i}) and the fact that $|G| \geq m- 2n$. 

Finally, we note that the bound also holds for the remaining agents $i$ that receive a bundle $B^t_i$ through the \textsc{MovingKnife} subroutine. As mentioned previously, at least $m-2n$ goods are passed on as input to the subroutine, i.e., if \textsc{MovingKnife} is executed on instance $\mathcal{J} = \langle G, {A}, \{v_i \}_{i \in A} \rangle$, then we have $|G| \geq m-2n$. Inequality (\ref{ineq:univ-l-i}) ensures that $\frac{1}{2n} v_i(G) \geq \ell_i$ for all $i \in A$.   Finally, using Lemma~\ref{lemma:movingknife}, we get that the bundle assigned to agent $i \in A$ satisfies the stated inequality: $v_i(B^t_i) = v_i (P_i) \geq \frac{1}{2n} v_i(G) \geq  \ell_i$. 

Hence, the stated claim follows. 
\end{proof}

We now show that the estimates $\gamma_i^t$s used in \textsc{Alg} also satisfy a lower bound similar to that in Lemma~\ref{lemma:lb-bundles}. 

\begin{lemma}
	\label{lemma:lb-gammas}
	Given a fair division instance $\mathcal{I} = \langle [m], [n], \{v_i\}_{i=1}^n \rangle$ with subadditive valuations, let $\gamma^t_i \in \mathbb{R}_+$ be the estimate associated with agent $i \in [n]$ in the $t^{\text{th}}$ iteration (for $t \in \mathbb{N}$) of the outer while loop (Step \ref{SATloop}) in \textsc{Alg}. Then,  for all agents $i \in [n]$ and each iteration count $t$, we have $\gamma_i^t \geq  \left(1 - \frac{1}{m} \right) \ell _i$.
%\begin{align*}
%	\gamma_i^t \geq  \left(1 - \frac{1}{m} \right) \ell _i.
%\end{align*}
\end{lemma}

\begin{proof} 
Note that for any agent $i \in [n]$, the quantity $\ell_i = 0$ iff $i$ has positive value for at most $2n$ goods. This observation implies that the initial for loop in \textsc{Alg} correctly identifies agents $i$ that have $\ell_i = 0$, and sets $\gamma^0_i = 0$. For such agents $\gamma^t_i =0$ for all $t$. %--\textsc{Alg} decrements $\gamma^t_i$s multiplicatively by a factor of $(1 - 1/m)$%.  
Hence, the lemma holds for any agent $i$ with $\ell_i = 0$.

We now consider agents $i\in [n]$ with $\ell _i > 0$. For such an agent $i$, the algorithm initially sets $\gamma^0_i = v_i([m])$. Hence, for $t=0$ we have $\gamma_i^t \geq \left(1 - \frac{1}{m} \right) \ell _i$. An inductive argument shows that this inequality continues to hold as the algorithm progresses. In particular, if in the $t^{\text th}$ iteration the algorithm does not decrement the estimate (i.e., if $i \in {\rm SAT}_{t+1}$), then $\gamma^{t+1}_i = \gamma^t_i \geq   \left(1 - \frac{1}{m} \right) \ell _i$. 

Even otherwise, if the algorithm multiplicatively decrements the estimate (in particular, sets $\gamma^{t+1}_i  = (1 - 1/m) \ \gamma^{t}_i$), then it must be the case that $ \gamma^{t}_i > v_i(B^{t}_i)$ (i.e., $i \in {\rm UNSAT}_{t+1}$). That is, after the decrement we have $\gamma^{t+1}_i \geq \left(1 - \frac{1}{m} \right) v_i(B^{t}_i) \geq \left(1 - \frac{1}{m} \right)  \ell_i$; the last inequality follows from Lemma~\ref{lemma:lb-bundles}. This completes the proof. 
\end{proof}

\subsection{Proof of Theorem~\ref{theorem:approximation-guarantee}} \label{section:main-proof}

In this section we prove Theorem~\ref{theorem:approximation-guarantee} by showing that \textsc{Alg} runs in polynomial time (Lemma~\ref{lemma:runtime}) and the computed allocation achieves the stated approximation ratio of $8n$ (Lemma~\ref{lemma:apx}). 

\begin{lemma}[Runtime Analysis]  \label{lemma:runtime}
Given any fair division instance $\mathcal{I}=\langle [m], [n], \{v_i\}_{i=1}^n \rangle$ in which the agents have monotone, subadditive valuations, \textsc{Alg} (Algorithm~\ref{Alg}) terminates after $T= {O} \left( n m \log \left ( nm  V \right) \right)$ iterations of its outer while loop (Step \ref{SATloop}); here, $V =\max \limits _{i \in [n]} \left ( \frac{ \max \limits _{g \in [m]} v_i(g)}{ \min \limits _{g \in [m]: v_i(g)>0} v_i(g)} \right)$.
%\begin{align*}
%V =\max \limits _{i \in [n]} \left ( \frac{ \max \limits _{g \in [m]} v_i(g)}{ \min \limits _{g \in [m]: v_i(g)>0} v_i(g)} \right ) \, .
%\end{align*}
\end{lemma}

\begin{proof}
By design, \textsc{Alg} iterates as long as ${\rm UNSAT}_t \neq \emptyset$. We will bound the number of times (i.e., the distinct values of $t$ for which) any agent $i\in [n]$ is contained in ${\rm UNSAT}_t$ and, hence, establish the stated runtime bound. 

Recall that for any agent $i \in [n]$, the quantity $\ell_i = 0$ iff $i$ has positive value for at most $2n$ goods. For such agents \textsc{Alg} sets $\gamma^0_i = 0$. Therefore, these agents are contained in ${\rm SAT}_t$, for all iterations $t \geq1$, and do not contribute to the repetitions of the outer while loop. 

For the remaining agents, with $\ell_i >0$, the algorithm initially sets $\gamma^0_i = v_i([m])$ and we have  
\begin{align}
\ell_i \geq \frac{1}{2n} \ \min \limits_{g \in [m]: v_i(g)>0} v_i(g) \, . \label{ineq:l-i-lower}
\end{align}

Using Lemma~\ref{lemma:lb-gammas} and the fact that the algorithm decrements $\gamma^{t}_i$ by a multiplicative factor of $(1-1/m)$ whenever $i \in {\rm UNSAT}_t$, we get that the number of times agent $i$ can be in the ${\rm UNSAT}_t$ is at most 
\begin{align*}
m \log \left( \frac{v_i([m])}{\ell_i } \right) & \leq m \log \left( \frac{ m \ \max_{g \in [m]} v_i(g) }{ \ell_i } \right) \tag{since $v_i$ is subadditive, $v_i([m]) \leq m \max_{g \in [m]} v_i(g)$} \\
& \leq m \log \left( \frac{ 2 n m \ \max_{g \in [m]} v_i(g) }{  \min \limits_{g \in [m]: v_i(g)>0} v_i(g) } \right) \tag{via inequality (\ref{ineq:l-i-lower})} \\
& \leq m \log \left(  2 nm  V \right). 
\end{align*}

Summing over all agents, we get that the number of times ${\rm UNSAT}_t \neq \emptyset$ is at most $T= {O} \left( n m \log \left ( nm  V \right) \right)$. Hence, the stated lemma follows. 
\end{proof}

We now show that the allocation computed by \textsc{Alg} achieves the required approximation guarantee. % stated in Theorem~\ref{theorem:approximation-guarantee}.

\begin{lemma} [Approximation Guarantee] \label{lemma:apx}

For any given fair division instance $\mathcal{I} = \langle [m], [n], \{v_i\}_{i=1}^n \rangle$ with subadditive valuations, let $\mathcal{B}=(B_1, \ldots, B_n)$ denote the allocation computed by \textsc{Alg}. Then, $\NSW (\mathcal{B}) \geq \frac{1}{8n} \NSW (\mathcal{N}^*)$; here, $\mathcal{N}^*$ denotes the Nash optimal allocation in $\mathcal{I}$. 
\end{lemma}

\begin{proof}
For the given instance $\mathcal{I}$, say \textsc{Alg} terminates after $T+1$ iterations of the outer while loop. That is, we have ${\rm UNSAT}_{T+1}=\emptyset$ and, for each agent $i \in [n]$, the returned bundle $B_i = B_{i}^{T} \cup \{ \pi^T (i) \}$. Here, $\pi^T (i)$ is the good assigned to agent $i$ under the maximum weight matching $\pi^T$ (considered in the last iteration) and $B_{i}^{T}$ is the bundle populated for $i$ (either in Step \ref{step:singleton} or in Step \ref{step:moving-knife}).  

The fact that ${\rm UNSAT}_{T+1}=\emptyset$ (i.e., ${\rm SAT}_{T+1} = [n]$) gives us 
\begin{align}
v_i(B_i^T) \geq \gamma_i^T  &\qquad \text{ for all } i \in [n] \, . \label{ineq:all-sat}
\end{align}

Lemma~\ref{lemma:upper-bound} (instantiated with $\mathcal{A}^* = \mathcal{N}^*$) implies that there exists a matching---${\sigma}(i) \coloneqq \widehat{g_i} \in N^*_i$, for all $i\in [n]$---with the property that $v_i(\sigma(i)) + \ell _i \geq \frac{1}{4n}v_i(N^*_i)$.
%\begin{align}
%v_i(\sigma(i))	+ \ell _i \geq \frac{1}{4n}v_i(N^*_i) \, .
%\end{align}
Using this inequality and Lemma \ref{lemma:lb-gammas} we get, for all $i \in [n]$:
\begin{align}
\label{ineq:gamma-upper}
v_i(\sigma(i)) + \gamma_i^T \geq \left ( 1- \frac{1}{m} \right )\frac{1}{4n}v_i(N^*_i) \, .
\end{align}

Recall that $\pi^T$ is a maximum weight matching in the bipartite graph (considered in Step \ref{step:matching} of \textsc{Alg}) with edge weights $ \log \left(v_i(g) + \gamma_i^T \right)$.  Given that $\sigma(\cdot)$ is some matching in the graph and $\pi^T$ is a maximum weight matching, we get  $\sum_{i=1}^n   \log \left( v_i(\pi^T(i)) + \gamma^T_i \right) \geq \sum_{i=1}^n \log \left( v_i(\sigma(i)) + \gamma^T_i \right)$. That is, 
\begin{align}
\left(\prod_{i=1}^n \left(v_i(\pi^T(i)) + \gamma_i^T \right) \right)^{\frac{1}{n}} & \geq \left(\prod_{i=1}^n \left(v_i( \sigma(i)) + \gamma_i^T \right) \right)^\frac{1}{n}  
\geq \left ( 1- \frac{1}{m}\right ) \frac{1}{4n} \NSW (\mathcal{N}^*) \, . \label{ineq:interim}
\end{align}
The last inequality follows from equation (\ref{ineq:gamma-upper}). Also, as defined previously, the optimal Nash social welfare $\NSW (\mathcal{N}^*) = \left(\prod_{i}^{n} v_{i} \left( N^*_{i}\right) \right)^{1/n}$.

The monotonicity of the valuation function $v_i$ implies $v_i \left(\{\pi^T(i) \} \cup B_i^T \right) \geq 1/2 \left(  v_i(\pi^T(i)) + v_i( B_i^T) \right)$ for each $i \in [n]$. Using these observations we can lower bound the Nash social welfare of the computed  allocation $\left( B_i = \{\pi^T(i) \} \cup B_i^T\right)_i $ as follows   
\begin{align*}
\left(\prod_{i=1}^n v_i( B_i )  \right)^{\frac{1}{n}} & \geq \frac{1}{2} \left(\prod_{i=1}^n \left(v_i(\pi^T(i)) + v_i(B^T_i) \right) \right)^{\frac{1}{n}} \\
& \geq \frac{1}{2} \left(\prod_{i=1}^n \left(v_i(\pi^T(i)) + \gamma_i^T\right) \right)^{\frac{1}{ n }} \tag{via inequality (\ref{ineq:all-sat})} \\
& \geq \left ( 1- \frac{1}{m} \right )\frac{1}{8 n} \NSW (\mathcal{N}^*) \, . \tag{via inequality (\ref{ineq:interim})}
\end{align*}
This establishes the stated approximation guarantee and completes the proof of the lemma. 
\end{proof}

\noindent
{\it Remark:} The result of Garg et al.~\cite{GargKK20} also holds for an asymmetric version of Nash social welfare maximization, in which each agent $i$ has an associated weight $\eta_i \geq 0$ and the goal is to find an allocation $(A_1, \ldots, A_n)$ that maximizes $\left( \prod_{i \in [n]} \left( v_i(A_i) \right)^{\eta_i} \right)^{\frac{1}{\sum_{i \in [n]} \eta_i}}$. Our approximation guarantee extends to this formulation. In particular, in Step \ref{step:matching} of $\textsc{Alg}$ we can set the edges weights to be $\eta_i \log (v_i(g) + \gamma_i)$ (instead of $ \log (v_i(g) + \gamma_i)$) and note that the subsequent arguments follow through to provide an $8n$-approximation ratio for maximizing Nash social welfare with asymmetric agents and subadditive valuations.    

Also, one can use Theorem \ref{theorem:approximation-guarantee}, in conjunction with the $m/n$ approximation guarantee of Nguyen and Rothe~\cite{NR14},\footnote{While Theorem 4 in \cite{NR14} provides the above-mentioned approximation guarantee of $(m-n+1)$, its proof can in fact be easily modified to obtain an approximation ratio of $m/n$.} to obtain an $O(\sqrt{m})$-approximation algorithm for maximizing Nash social welfare under subadditive valuations: for instances in which $m \geq n^2$, the $8n$ approximation suffices. Otherwise, if $m < n^2$ (i.e., $m/n < \sqrt{m}$), then we can invoke the result of Nguyen and Rothe~\cite{NR14}. 

\section{An $8n$-Approximation for $p$-Mean Welfare}
\label{section:p-mean}
This section shows that we can extend Algorithm \ref{Alg} and obtain an $8n$ approximation for maximizing the $p$-mean welfare as well. % for any $p \in (-\infty, 1]$. 

%\subsection{Algorithm and Main Result}

For maximizing $p$-mean welfare, \textsc{ALG} (Algorithm \ref{Alg}) is modified as follows: In Step ~\ref{step:matching}, the weight $w(i,g)$ of edge $(i,g) \in [n]\times [m]$ is set as $(v_i(g) + \gamma _i^t)^p$ (instead of $ \log \left(v_i(g) + \gamma _i^t\right)$).\footnote{Recall that the $p =0$ case corresponds to Nash social welfare. Since we already have the desired approximation guarantee for this case, it is not explicitly addressed in this section.} Furthermore, 
\begin{itemize} 
\item[(i)] For $p \in (0,1]$, in Step \ref{step:compute-matching} we compute a left-perfect \emph{maximum}-weight matching, $\pi^t$, otherwise 
\item[(ii)] For finite $p < 0$, we compute a left-perfect \emph{minimum}-weight matching, $\pi^t$, in Step \ref{step:compute-matching}
\item[(iii)] For maximizing egalitarian welfare (the $p = -\infty$ case), we set edge weights to be $\left( v_i(g) + \gamma _i^t \right)$ and compute a max-min matching\footnote{In particular, via binary search (over edge weights), we find a matching wherein the minimum edge weight (across agents) is as high as possible.} $\pi^t$ with respect to these weights.
\end{itemize}

%In phase one of the while loop \ref{SATloop},  computes a maximum weight matching with agent-goods pair weights set to $ \log \left(v_i(g) + \gamma _i^t\right) $. In order to extend \textsc{ALG} for weighted $p$-means, given a $p \in (-\infty ,1]$, we define the agent-good weights as $\eta _i (v_i(g) + \gamma _i^t)^p$ instead while computing the matching. For $p>0,$ phase one then computes a maximum weight matching over the agent-good pairs, while phase two and three remain the same as before. In contrast, for $p<0,$ phase one computes a minimum weight (left-perfect) matching over the agent goods pairs and rest of the phases remain same. Once again we run the three phases iteratively until all agents receive a bundle of value at least $\gamma _i^t$. 

Theorem~\ref{theorem:approximation-guarantee-for-p} below establishes that, with these changes in \textsc{ALG} (Algorithm \ref{Alg}), we can efficiently compute an allocation with $p$-mean welfare at least $\frac{1}{8n}$ times the optimal ($p$-mean welfare). Note that by Proposition~\ref{prop:large-p}, for $p \le - n \log n$, we can maximize the egalitarian welfare, instead of the $p$-mean welfare, and the allocation thus obtained is an $8n$-approximation to the optimal $p$-mean welfare allocation.

%When $p=-\infty,$ the edge weight $\eta _i (v_i(g) + \gamma _i^t)^p$ is infinite, and hence cannot be computed. Therefore, we exploit the fact that when $p=- \infty$, the generalized $p$-mean of an allocation $(A_1,\dots , A_n)$ is $\min \limits _{i\in [n]} (v_i(A_i)) $. Hence, we set edge weights to $\left(v_i(g) + \gamma _i^t\right )$ and compute a max-min matching with respect to these weights. The following algorithm, \textsc{Alg for $p$-mean} \ref{Alg_p}, computes an allocation that is an $\mathcal{O}(n)$-approximation to the optimal $p$-mean.

\begin{theorem} \label{theorem:approximation-guarantee-for-p}	
Let $\mathcal{I} = \langle [m], [n], \{v_i\}_{i=1}^n \rangle$ be a fair division instance in which the valuation function $v_i$, of each agent $i \in [n]$, is nonnegative, monotone, and subadditive. Then, given value oracle access to $v_i$s , one can efficiently compute an $8 n$ approximation to the optimal $p$-mean welfare for any $p \in (-\infty, 1]$.
\end{theorem}

\begin{proof}
We first note that Lemmas \ref{lemma:upper-bound}, \ref{lemma:movingknife}, \ref{lemma:lb-bundles}, \ref{lemma:lb-gammas}, and \ref{lemma:runtime} hold as is for $p$-mean welfare. In particular, using Lemma \ref{lemma:runtime} we get that, even with the above-mentioned changes, the algorithm runs in polynomial time. 

To complete the proof of the theorem, we will next show that the computed allocation $\mathcal{B}=(B_1, \ldots, B_n)$ satisfies ${\rm M}_p(\mathcal{B}) \geq \frac{1}{8n} {\rm M}_p (\mathcal{A}^*(p))$, where $\mathcal{A}^*(p)$ %=(A^*_1(p), \ldots, A^*(p))$ 
is a $p$-mean welfare maximizing allocation.

For the given instance $\mathcal{I}$, say the modified algorithm terminates after $T+1$ iterations of the outer while loop. That is, we have ${\rm UNSAT}_{T+1}=\emptyset$ and, for each agent $i \in [n]$, the returned bundle $B_i = B_{i}^{T} \cup \{ \pi^T (i) \}$. Here, $\pi^T (i)$ is the good assigned to agent $i$ under the matching $\pi^T$ (considered in the last iteration) and $B_{i}^{T}$ is the bundle populated for $i$ in the final iteration. 

The fact that ${\rm UNSAT}_{T+1}=\emptyset$ (i.e., ${\rm SAT}_{T+1} = [n]$) gives us 
\begin{align}
v_i(B_i^T) \geq \gamma_i^T  &\qquad \text{ for all } i \in [n] \label{ineq:all-sat-p}
\end{align}

Lemma~\ref{lemma:upper-bound} implies that there exists a matching---${\sigma}(i) \coloneqq \widehat{g_i} \in A^*_i(p)$, for all $i\in [n]$---with the property that $v_i(\sigma(i))	+ \ell _i \geq \frac{1}{4n}v_i(A^*_i(p))$.
Using this inequality and Lemma \ref{lemma:lb-gammas} we get, for all $i \in [n]$:
\begin{align}
\label{ineq:gamma-upper-p}
v_i(\sigma(i)) + \gamma_i^T \geq \left ( 1- \frac{1}{m} \right )\frac{1}{4n}v_i(A^*_i(p))
\end{align}

Recall that, given $p$, the modified algorithm computes $\pi ^T$ based on the sign of $p$. Hence, we split the proof of Theorem \ref{theorem:approximation-guarantee-for-p} into three cases depending on whether $p>0$, $p<0$, or $p=-\infty$. \\

\noindent {Case (i)}: $p>0$. In this case, $\pi^T$ is a left-perfect maximum-weight matching in the bipartite graph $([n] \cup [m], [n] \times [m])$ with edge weights $ \left(v_i(g) + \gamma_i^T \right)^p$.  Given that $\sigma(\cdot)$ is some (left-perfect) matching in the graph and $\pi^T$ is a maximum-weight matching, we get  $\sum_{i=1}^n  \left( v_i(\pi^T(i)) + \gamma^T_i \right)^p \geq \sum_{i=1}^n  \left( v_i(\sigma(i)) + \gamma^T_i \right)^p$. Therefore, with $p>0$, the following inequality holds
\begin{align}
\left(\frac{1}{n} \sum_{i=1}^n  \left(v_i(\pi^T(i)) + \gamma_i^T\right)^p \right)^{\frac{1}{p}} & \geq \left(\frac{1}{n} \sum_{i=1}^n    \left(v_i( \sigma(i)) + \gamma_i^T \right)^p\right)^\frac{1}{p} \geq \left( 1- \frac{1}{m}\right) \frac{1}{4n} {\rm M}_p (\mathcal{A}^*(p)) \, . \label{ineq:interim-p-pos}
\end{align}
The last inequality follows from~\eqref{ineq:gamma-upper-p}. \\ %Recall that the optimal $p$-mean welfare ${\rm M}_p \left( \mathcal{A}^*(p) \right) = \left( \sum_{i=1}^{n} v_{i}\left(A^*_{i}(p)\right)^p\right)^\frac{1}{p}$. \\

\noindent {Case (ii)}: Finite $p<0$.  By design, in this case, $\pi^T$ is a left-perfect minimum-weight matching in the bipartite graph $([n] \cup [m], [n] \times [m])$ with edge weights $ \left(v_i(g) + \gamma_i^T \right)^p$.  Given that $\sigma(\cdot)$ is some left-perfect matching in the graph and $\pi^T$ is a minimum-weight matching, we get $\sum_{i=1}^n   \left( v_i(\pi^T(i)) + \gamma^T_i \right)^p \leq \sum_{i=1}^n  \left( v_i(\sigma(i)) + \gamma^T_i \right)^p$. The fact that $p$ is negative gives us 
\begin{align}
\left(\frac{1}{n} \sum_{i=1}^n  \left(v_i(\pi^T(i)) + \gamma_i^T\right)^p \right)^{\frac{1}{p}} & \geq \left(\frac{1}{n} \sum_{i=1}^n  \left(v_i( \sigma(i)) + \gamma_i^T \right)^p\right)^\frac{1}{p} \geq \left ( 1- \frac{1}{m}\right ) \frac{1}{4n} {\rm M}_p (\mathcal{A}^*(p)) \label{ineq:interim-p-neg}
\end{align}
The last inequality follows from~\eqref{ineq:gamma-upper-p}. \\

\noindent {Case (iii)}: $p=-\infty$. In this case, $\pi ^T$ is a max-min matching computed with edge weights $ \left( v_i(\pi^T(i)) + \gamma^T_i \right)$. Given that $\sigma(\cdot)$ is some matching in the graph and matching $\pi^T$ maximizes the value of the minimum matched edge, we get $\min _{i\in [n]} (v_i(\pi ^T(i))+ \gamma _i^T) \geq \min _{i\in [n]} (v_i(\sigma ^T(i))+ \gamma _i^T)$. Therefore, 

\begin{align}
 \min _{i\in [n]} (v_i(\pi ^T(i))+ \gamma _i^T) &\geq \min _{i\in [n]} (v_i(\sigma ^T(i))+ \gamma  _i^T)  \geq \left ( 1- \frac{1}{m}\right ) \frac{1}{4n} {\rm M}_p (\mathcal{A}^*(p)) \label{ineq:interim-p-inf}
\end{align}
The last inequality follows from~\eqref{ineq:gamma-upper-p}.

The monotonicity of the valuation function $v_i$ implies $v_i \left(\{\pi^T(i) \} \cup B_i^T \right) \geq 1/2 \left(  v_i(\pi^T(i)) + v_i( B_i^T) \right)$ for each $i \in [n]$. Using these observations we can lower bound the $p$-mean welfare of the computed  allocation $\left( B_i = \{\pi^T(i) \} \cup B_i^T\right)_i $ as follows   
\begin{align*}
\left(\sum_{i=1}^n   \left(v_i( B_i )\right)^p\right)^\frac{1}{p} & \geq \frac{1}{2} \left(\sum_{i=1}^n  \left(v_i(\pi^T(i)) + v_i(B^T_i) \right)^p \right)^\frac{1}{p} \\
& \geq \frac{1}{2} \left(\sum_{i=1}^n   \left(v_i(\pi^T(i)) + \gamma_i^T\right)^p\right)^\frac{1}{p} \tag{via inequality (\ref{ineq:all-sat-p})} \\
& \geq \left ( 1- \frac{1}{m} \right )\frac{1}{8 n} {\rm M}_p (\mathcal{A}^*(p)) \tag{via inequality (\ref{ineq:interim-p-pos}), (\ref{ineq:interim-p-neg}), or (\ref{ineq:interim-p-inf})}
\end{align*}
This establishes the stated approximation guarantee and completes the proof of the theorem. 
\end{proof}

\section{Lower Bound on Approximating $p$-Mean Welfare}
\label{section:lower-bound}

This section shows that, under XOS valuations, maximizing the $p$-mean welfare for $p \in (-\infty,1]$ within a sub-linear (in $n$) approximation factor necessarily requires an exponential number of value queries (Theorem \ref{theorem:query}). This result directly implies that the approximation ratio obtained in Theorems \ref{theorem:approximation-guarantee} and \ref{theorem:approximation-guarantee-for-p} (via polynomially many value queries) is essentially tight. We note that this query lower bound is unconditional, i.e., it does not depend on any complexity theoretic assumption.% and it also holds for the $p$-mean welfare maximization problem, with $p \in (-\infty, 1]$.     

We establish Theorem \ref{theorem:query} by directly adapting a result of Dobzinski et al.~\cite{DobzinskiNS10}, which provides a similar lower bound for social welfare. The impossibility result here holds under {XOS} valuations; \footnote{Our results work under the value oracle model and do not require an explicit description of the underlying additive functions that define the XOS function at hand.} recall that XOS valuations constitute a special class of subadditive functions.%Specifically, a set function, $f:2^{[m]} \mapsto \mathbb{R}_+$, is said to be XOS iff it is obtained by evaluating the maximum over a collection of additive functions $\{a_j \}_{j \in [L]}$, i.e., $f(S) \coloneqq \max_{1 \leq j \leq L} \left\{ a_j(S) \right\}$, for each subset $S \subset [m]$.\footnote{Here, $L$ can be exponentially large in $m$.

 %and XOS valuations. In order to obtain such a result for social welfare, the authors first consider an instance with $m$ goods and $\sqrt{m}$ agents. They then define two distinct XOS valuation functions $v_i$ and $v_i'$ for each agent $i$ such that the optimal social welfare when the valuations are $v'_i$s is $\mathcal{O}(n)$ times  the optimal social welfare when the valuations are $v_i$s. Finally, by showing that it takes exponentially many value queries to distinguish between the above two cases, they obtained the desired result. We merely redefine the given instance to have $n$ agents and $n^2$ goods, instead of $\sqrt{m}$ agents and $m$ goods. The rest of the proof follows through by replacing all occurances of $m^{1/2}$ with $n$. For the sake of completeness and ease of understanding, we include the proof of Theorem \ref{hardness}.

\begin{restatable}{theorem}{TheoremQuery}
\label{theorem:query}
For fair division instances $\mathcal{I}=\langle [m],[n], \{v_i \}_{i=1}^n \rangle$ with XOS valuations and $p \in (-\infty,1]$, finding an allocation with $p$-mean welfare at least ${1}/{ n^{1- \varepsilon} }$ times the optimal requires exponentially many value queries; here $\varepsilon >0$ is any fixed constant. 
\end{restatable}

%\begin{theorem}\label{theorem:query}

%\end{theorem}

%Our construction follows that of Dobzinski et al., hence
Here, we briefly explain the salient points of the proof of this lower bound and provide the details in Appendix \ref{appendix:query}. Dobzinski et al.~\cite{DobzinskiNS10} construct two (families of) instances, both with $n$ agents, $m=n^2$ goods, and XOS valuations for the agents. In the first instance, each agent has the same valuation function $f: 2^{[m]} \mapsto \mathbb{R}_+$ and maximum average social welfare ($1$-mean welfare) is $ n^{4\delta}$, for a fixed constant $\delta>0$. In the second instance, each agent has her own (non-identical) valuation function $v_i: 2^{[m]} \mapsto \mathbb{R}_+$ and there exists an allocation in which each agent has value $n$ for her bundle. For any $p \le 1$, it follows that in the first instance the optimal $p$-mean welfare is at most $n^{4 \delta}$ (via the generalized mean inequality), while for the second instance, the optimal $p$-mean welfare is at least $n$ (since there exists an allocation where every agent achieves value $n$). The proof of Dobzinski et al.~\cite{DobzinskiNS10} goes on to show that it takes an exponential number of value queries to distinguish between the two instances. However, given an $O(n^{1-\varepsilon})$-approximation algorithm for the $p$-mean welfare, one can readily distinguish between the two instances (by choosing $\delta < \varepsilon/4$). Hence such an algorithm must make an exponential number of value queries.

%and there exists an allocation of goods to agents where each agent has value $n$ for the bundle she receives. In the second instance, each agent has her own (non-identical) valuation function $v_i: 2^{[m]} \mapsto \mathbb{R}_+$, and the maximum average social welfare (i.e., for $p=1$) obtainable is $n^{4 \delta}$ for any fixed constant $\delta$. We note that for any $p \le 1$, it follows that in the first instance, the optimal $p$-mean welfare is at least $n$ (since there exists an allocation where every agent has value $n$), while for the second instance, the optimal $p$-mean welfare is at most $n^{4 \delta}$ (by the generalized mean inequality). The proof of Dobzinski et al. goes on to show that it takes an exponential number of oracle value queries to distinguish between the two instances. However given an $O(n^{1-\epsilon})$-approximation algorithm for the $p$-mean welfare, one can readily distinguish between the two instances (by choosing $\delta < \epsilon/4$). Hence such an algorithm must make an exponential number of oracle value queries.

\section{$(m-n+1)$-Approximation Guarantees}
\label{section:m-n}

This section provides two extensions of the result of Nguyen and Rothe~\cite{NR14}, which shows the Nash social welfare maximization problem (under subadditive valuations) admits an $(m-n+1)$-approximation algorithm. First, we show that an $(m-n+1)$-approximation for the $p$-mean welfare can be obtained for all $p \le 0$ and with subadditive valuations. Then, we establish that it is {\rm NP}-hard to extend this positive result to any $p \in (0,1)$, even under additive valuations, i.e., it is {\rm NP}-hard to obtain an $(m-n+1)$-approximation for $0< p < 1$. The proofs of these two results are deferred to Appendix \ref{appendix:m-n}.

%for $0< p < 1$, we establish that it is {\rm NP}-hard to obtain an $(m-n+1)$-approximation, even under additive valuations.
% provides an $(m-n+1)$-approximation algorithm for maximizing Nash social welfare under subadditive valuations. %\footnote{The result of Nguyen and Rothe can be easily modified to obtain an $(m/n)$-approximation for subadditive valuations, by observing that in their proof, all inequalities continue to hold for subadditive valuations, and if $(m-n+1)$ is replaced by $(m/n)$.} 
%Firstly, an $(m-n+1)$-approximation to the $p$-mean welfare can be obtained for all $p \le 0$, for subadditive valuations. Secondly, it is NP-hard to extend this to $p<0$, even for additive valuations.

\begin{restatable}{theorem}{TheoremMNupper}
\label{theorem:m-n-upper}
Let $\mathcal{I} = \langle [m], [n], \{v_i\}_{i=1}^n \rangle$ be a fair division instance in which the valuation function $v_i$, of each agent $i \in [n]$, is nonnegative, monotone, and subadditive. Then, given value oracle access to $v_i$s, one can efficiently compute an $(m-n+1)$ approximation to the $p$-mean welfare maximization problem for any $p \in (-\infty, 0]$.
\end{restatable}

%polynomially bounded

The next theorem asserts that it is unlikely that Theorem \ref{theorem:m-n-upper} extends to $p \in (0,1)$.%, even for additive valuations.

\begin{restatable}{theorem}{TheoremMNlower}
\label{theorem:m-n-lower}
For fair division instances  $\mathcal{I} = \langle [m], [n], \{v_i\}_{i=1}^n \rangle$ with additive valuations and for any fixed $p \in (0,1)$, computing an allocation with $p$-mean welfare at least $1/(m-n+1)$-times the optimal (for all $m$ and $n$) is {\rm NP}-hard.
\end{restatable}
Note that this hardness result (in light of Theorem \ref{theorem:approximation-guarantee-for-p}) is relevant for instances in which $m < 2n$. 
%By a very similar proof, 
%We can also show that when agents have submodular valuations, it is {\rm NP}-hard to approximate the optimal social welfare (the $p=1$ case) by a factor better than $\frac{e}{e-1}(m-n+1)$. This inapproximability is obtained via a reduction from the hardness result of Khot et al.~\cite{KhotLMM08}. We defer the details to a full version of the paper. 

\section*{Acknowledgements}
Siddharth Barman gratefully acknowledges the support of a Ramanujan Fellowship (SERB - {SB/S2/RJN-128/2015}) and a Pratiksha Trust Young Investigator Award. Umang Bhaskar's research is generously supported the Department of Atomic Energy, Government  of India (project no. RTI4001), a Ramanujan Fellowship (SERB - SB/S2/RJN-055/2015), and an Early Career Research Award (SERB - ECR/2018/002766).

\bibliographystyle{alpha}
\bibliography{subadditive}

\appendix

\section{Appendix}

\subsection{Missing Proof from Section~\ref{section:notation}}
\label{app:notation}
This section restates and proves Proposition \ref{prop:large-p}.

\PropMp*
\begin{proof}
The first inequality $\M_{-\infty}(x_1, \ldots, x_n) ~\le ~ \M_{p}(x_1, \ldots, x_n)$ is an instantiation of the generalized mean inequality. 
%, which states that for any $r < s$, $\M_r(\mathcal{X}) \le \M_s(\mathcal{X})$, with equality only if all $v_i$s are equal. 
For the second bound $\M_{p}(x_1, \ldots, x_n) ~\le ~2^{1/n} \ \M_{-\infty}(x_1, \ldots, x_n)$, we assume, without loss of generality, that $x_1 = \min_i x_i$. Note that $(x_i/x_1)^p \ge 0$ for all $i \in [n]$ and, hence, we have $\left(1 + \left(\frac{x_2}{x_1}\right)^p + \ldots + \left(\frac{x_n}{x_1}\right)^p\right)\geq 1$. Therefore, the fact that $p$ is negative gives us 
\begin{align}
\left(1+ \left(\frac{x_2}{x_1}\right)^p + \ldots + \left(\frac{x_n}{x_1}\right)^p\right)^{1/p} & \le 1^{1/p} = 1 \label{equation:p-ineq}
\end{align}

This inequality leads to the stated upper bound 
\begin{align*}
\M_p(x_1, \ldots, x_n) & = \left( \frac{1}{n}(x_1^p + \ldots + x_n^p) \right)^{1/p} \\
	& = \frac{x_1}{n^{1/p}}  \left(1 + \left(\frac{x_2}{x_1}\right)^p + \ldots + \left(\frac{x_n}{x_1}\right)^p \right)^{1/p} \\
	& \le \frac{x_1}{n^{1/p}}   \tag{from equation~\eqref{equation:p-ineq}} \\
	%& = x_1 \left( \frac{1}{n} \right)^{1/p} \\
	& \le x_1 \ n^{1/(n \log n)}  \tag{since $p \leq -n \log n$} \\ &  = ~2^{1/n} x_1 ~= ~ 2^{1/n} \M_{-\infty}(x_1, \ldots, x_n) \, .
\end{align*}
\end{proof}
	
\begin{comment}
In our modified algorithm to approximate the optimal $p$-mean welfare, in the first step when we compute a matching, we set the weight of an edge $(i,g)$ to $(v_i(g) + \gamma_i^t)^p$. The weight of an edge thus requires $p \log (v_i(g) + \gamma_i^t)$ bits to represent. If the absolute value of $p$ is exponentially large, then the weight of the edges require exponentially many bits to represent. To fix this issue, if $p \le -n \log n$, we run the algorithm assuming that $p = - \infty$ (i.e., assume the objective is the egalitarian welfare) instead.

To show that it is sufficient to maximize the egalitarian welfare for $p \le -n \log n$, let $\mathcal{A}^*$ be an optimal allocation for the $p$-mean welfare, $\mathcal{B}^*$ be an optimal allocation for the egalitarian welfare, and $\mathcal{B}$ be an $8n$-approximation to the egalitarian welfare. Then from Proposition~\ref{prop:large-p} it follows that:

\begin{align*}
\M_p(\mathcal{B}) & \ge \M_{-\infty}(\mathcal{B}) \tag{Proposition~\ref{prop:large-p}} \\ 
	& \ge \frac{1}{8n}\M_{-\infty}(\mathcal{B}^*) \tag{$\mathcal{B}$ is an $8n$-approximation} \\
	& \ge \frac{1}{8n}\M_{-\infty}(\mathcal{A}^*) \tag{$\mathcal{B}^*$ is an optimal egalitarian welfare allocation} \\
	& \ge \frac{1}{16n} \M_p(\mathcal{B}^*) \tag{Proposition~\ref{prop:large-p}} \\
\end{align*}
\end{comment}

\subsection{Missing Proofs from Section~\ref{section:lower-bound}}
\label{appendix:query}
Here we establish Theorem~\ref{theorem:query}.

\TheoremQuery*
\begin{proof}
Consider a family of instances with $n$ agents and $m=n^2$ goods. In addition, for a subset of goods $S \subseteq [m]$, define an additive function $a_S (\cdot)$ as follows: $a_S(g) \coloneqq 1$ if $g \in S$, otherwise (if $g \notin S$) we have $a_S(g) \coloneqq 0$. 
%\begin{align*}
%a_S(g) & = \begin{cases}
%1 \text{ if  $g\in S$}\\
%0 \text{ otherwise}
%\end{cases}
%\end{align*}
That is, $a_S(T) = |S \cap T|$ for all $T \subseteq [m]$. Furthermore, for an arbitrarily small constant $\delta >0$, define the additive function $\overline{a} (\cdot)$ as $\overline{a}(g) \coloneqq \frac{1+\delta}{n^{1-2 \delta}}$ for each good $g \in [m]$.

We now construct an XOS function
\begin{align}
f(T) \coloneqq  \max \left\{ \max \limits _{S \subseteq [m]: |S| \leq (1+ \delta)n^{4\delta}} a_S(T), \ \overline{a}(T) \right\} \qquad \text{for all $T \subseteq [m]$} \label{defn:v}
\end{align}

For any subset of goods $T$, if $|T|\leq (1+ \delta)n^{4\delta}$, then we have $f (T)=a_T(T)=|T|$. This equality follows from the fact that, in this cardinality range, $a_T(T)$ is strictly greater than $\overline{a}(T)$. Also, note that for sets of size more than $(1+ \delta)n^{4\delta}$, the values of the additive functions $\{ a_S \}_{S}$ considered in (\ref{defn:v}) plateau at $(1+ \delta)n^{4\delta}$. However, it is only when $|T|>n^{1+ 2\delta}$, that the term $\overline{a}(T)$ dominates $a_S(T)$. %Indeed, $f(T)$ depends only on the cardinality of $T$. 

Using function $f$, we define an XOS valuation $v_i$ for each agent $i \in [n]$. Select a partition $T_1,\ldots, T_n$ of the $m=n^2$ goods uniformly at random such that $|T_i| =n$ for each $i \in [n].$ Then, valuation ${v}_i$ is defined as follows 
\begin{align}
{v}_i(T) \coloneqq \max \left\{f (T), a_{T_i}(T) \right\} \qquad \text{for all $T \subseteq [m]$} \label{defn:v-i}
\end{align}

To prove the stated query lower bound, we consider two families of instances. One in which the valuation of each agent is $f$ and the other in which the agents' valuations are $v_i$s. In Claim \ref{claim:exp-queries} below we prove that an exponential number of value queries are required to differentiate between these two cases, i.e., to determine whether an agent $i$'s valuation is $f$ or $v_i$. 

However, this distinction can be made via an $n^{1 - \varepsilon}$ approximation to the optimal $p$-mean welfare: if the valuation of each agent is $f$, then the maximum average social welfare ($1$-mean welfare)  is ${O}\left ( n^{4\delta}\right )$ and, hence, by monotonicity of $p$-mean welfare (with $p \leq 1$) we get that the optimal $p$-mean welfare is also ${O}\left ( n^{4\delta}\right )$ in this case. By contrast, when the valuation functions are $v_i$s, by assigning subset $T_i$ to each agent $i$, we can ensure that each agent receives a bundle of value $n$. That is, in this setting, the optimal $p$-mean welfare is equal to $n$.  This gap between the optimal $p$-mean welfare in the two cases implies that, with a sub-linear approximation in hand, one can distinguish between $f$ and $v_i$. However, the following claim shows that this task requires an exponential number of value queries. 

%We now prove that it takes an exponential number of queries to distinguish between the cases when the agent $i$'s valuation is $v_i$ and the case that their valuation function is $\overline{v}_i$.  Note that the maximum average social welfare attainable when the valuation functions are $v_i$s is $\mathcal{O}\left ( n^{4\delta}\right )$. This is because for every agent $i\in [n]$ and for any set $T\subset [n^2]$, we have $v_i(T)\leq (1+\delta)n^{4\delta}$.  Therefore, by monotonicity of $p$-means with respect to $p$, the maximum $p$-mean attainable is also $\mathcal{O}\left ( n^{4\delta}\right )$, for any $p\leq 1.$ In contrast, when the valuation functions are $\overline{v}_i$s, the maximum average social welfare attainable is $n$. Note that this value is attained when each agent $i$ gets $T_i$, of value $n$ for agent $i$. Additionally, this means the maximum attainable $p$-mean is also $n$. Suppose we require exponentially many value queries to distinguish between $v_i$ and $\overline{v}_i$. Then, in order to obtain a $\Omega \left (n^{1-4\delta} \right )$-approximation ratio for maximum $p$-mean welfare, one needs exponentially many value queries.

\begin{claim} \label{claim:exp-queries} 
An exponential number of value queries are required to distinguish whether an agent $i$'s valuation is $f$ or $v_i$. 
\end{claim}

\begin{proof} We will prove that, for any subset $S \subseteq [m]$, the inequality $v_i(S) \neq f(S)$ holds with exponentially small probability.\footnote{Recall that the partition $T_1, \ldots, T_n$ of the $m$ goods is selected at random.} Hence, an exponential number of value queries are required to distinguish between these two functions.  
 
Note that, for a subset $S$, $v_i(S) \neq f(S)$ iff $a_{T_i}(S) > f(S)$ (see~\eqref{defn:v-i}). The following cases identify conditions (on subsets $S$) under which we have $a_{T_i}(S) > f(S)$. \\

\noindent {Case 1:} $|S|\leq (1+\delta)n^{4\delta}$. For subsets with this small a cardinality, functions $v_i$ and $f$ have the same value. In particular, if $|S|\leq (1+\delta)n^{4\delta}$, then $f(S) = |S| \geq |S \cap T_i| = a_{T_i}(S)$ and, hence, the equality $v_i(S) = f(S)$ holds.  \\

\noindent {Case 2:} $(1+\delta)n^{4\delta} < |S| \leq  n^{1+2\delta}$. As observed previously, in this cardinality range, $f(S) = (1+\delta)n^{4\delta}$. Therefore, for the inequality $a_{T_i}(S) > f(S)$ to hold, we require $a_{T_i}(S) = |S \cap T_i| > (1+ \delta) n^{4\delta}$.  

Since the partition $T_1, \ldots, T_n$ of $[m]$ (with $m = n^2$ and $|T_j|= n$ for each $j$) was chosen uniformly at random, we have $\mathbb{E} \left[ |S \cap T_i| \right] = |S|/n \leq n^{4 \delta}$. Applying Chernoff bounds, we get $\Pr \left\{ |S \cap T_i| \geq (1+ \delta) n^{4\delta} \right\} \leq \mathrm{exp} \left( - \frac{n^{4\delta} \delta^2 }{3} \right)$. Therefore, in this case, $v_i(S) \neq f(S)$ with exponentially small probability. \\

\noindent {Case 3:} $|S|> n^{1+2\delta}$. Here, $f(S) = \overline{a} (S) = \frac{(1+ \delta)}{n^{1-2\delta}} |S|$ (see equation (\ref{defn:v})). Therefore, the inequality $a_{T_i}(S) > f(S)$ holds iff $|S \cap T_i| > (1+ \delta) \frac{|S|}{n^{1-2\delta}} $. Again, via Chernoff bound, we get that $\Pr \left\{ |S \cap T_i| > (1+ \delta) \frac{|S|}{n^{1-2\delta}} \right\}$ is exponentially small. 

Overall, these observations show that, irrespective of the size of the $S$, the separation $v_i(S) \neq f(S)$ holds with exponentially small probability. This establishes the claim. 
\end{proof}

As mentioned previously, Claim \ref{claim:exp-queries} implies that an exponential number of value queries are required to approximate the maximum $p$-mean welfare within a factor of $n^{1 - \varepsilon}$.  
\end{proof}

\subsection{Missing Proofs from Section~\ref{section:m-n}}
\label{appendix:m-n}

In this section we restate and prove Theorems \ref{theorem:m-n-upper} and \ref{theorem:m-n-lower}. 

\TheoremMNupper*
\begin{proof} 
We actually show that an optimal matching achieves the stated approximation bound. Consider the bipartite graph $ \left( [n] \cup [m],  [n]  \times [m], \{ w(i,g) \}_{i \in [n], g \in [m]} \right)$ with weight of edge $(i,g) \in [n] \times [m] $ set as $w(i,g) =  \left(v_i(g)\right)^p$ (for $p=0$, we set the weight to be $\log v_i(g)$). Compute a left-perfect \emph{minimum}-weight matching, $\pi$, in this bipartite graph. Assign the remaining items arbitrarily to the agents and let $\mathcal{P} = (P_1,\ldots ,P_n)$ be the resulting allocation. We will show that $\mathcal{P}$ obtains the required approximation ratio. 

For the given $p \le 0$, let $\mathcal{A}^* = (A_1^*,\ldots ,A_n^*)$ denote a $p$-mean welfare maximizing allocation. Since $p \le 0$, in allocation $\mathcal{A}^*$ each agent is allocated at least one good (otherwise the optimal value is zero). Hence, any agent is allocated at most $m-n+1$ goods, $|A^*_i| \leq m-n+1$ for all $i \in [n]$. Let $g_i^*$ be the highest valued (by $i$) good in $A^*_i$, i.e., $g^*_i = \argmax_{g\in A^*_i} \ v_i(g)$. Then, the subadditivity of $v_i$ implies that 
\begin{align}
v_i(g_i^*) \geq \frac{1}{m-n+1} v_i(A_i^*)  \, .\label{eqn:m-n-matching}
\end{align}

\noindent Further, since allocating $g_i^*$ to agent $i$ constitutes a feasible matching of the goods to agents and $\pi$ is a minimum-weight matching, we have $\sum_{i \in [n]} v_i(\pi(i))^p \le \sum_{i \in [n]} v_i(g_i^*)^p$. The fact that $p$ is negative gives us $ \left( \sum_{i}  v_i(\pi(i))^p \right)^{1/p}  \geq \left(  \sum_{i}  v_i(g_i^*)^p \right)^{1/p}$. Hence, 

\begin{align*}
\left( \frac{1}{n} \sum_{i \in [n]} \left(v_i(P_i)\right)^p \right)^{1/p} & \ge \left( \frac{1}{n} \sum_{i \in [n]} \left(v_i(\pi(i))\right)^p \right)^{1/p} 
	~ \ge ~\left( \frac{1}{n} \sum_{i \in [n]} \left(v_i(g_i^*)\right)^p \right)^{1/p}  
	~ \ge ~ \frac{1}{m-n+1} \left( \frac{1}{n} \sum_{i \in [n]} \left(v_i(A_i^*)\right)^p \right)^{1/p} \, .
	\end{align*}

The last inequality follows from equation~\eqref{eqn:m-n-matching}. This completes the proof.
\end{proof}

\TheoremMNlower*
\begin{proof}
Our reduction is from the \textsc{Partition} problem: given a set $S = \left\{s_i \in \mathbb{Z}_+ \right\}_{i=1}^m$ of $m$ positive integers, determine if there exists a subset $T \subseteq [m]$ of indices such that $\sum_{i \in T} s_i =\frac{1}{2}\sum_{i \in [m]} s_i$. \textsc{Partition} is one of the classic {\rm NP}-hard problems.

Given an instance of \textsc{Partition}, we construct a fair division instance with additive valuations as follows. Write $z \coloneqq \sum_{i \in [m]} s_i$ and consider $n=m$ agents and $m$ goods, $g_1, \ldots, g_m$. The first two agents have value $s_i$ for good $g_i$, for each $i \in [m]$; hence, these two agents have identical additive valuations. The remaining $m-2$ agents have value $0$ for all goods.\footnote{Instead of zero, we could also set $v_i(g) = \left((z+s_i)^p - z^p\right)/2$ for these agents and all goods. In any optimal allocation, all goods must then be allocated to the first two agents.}  Note that in this instance, since the number of goods and number of agents is equal, an $(m-n+1)$-approximation algorithm must in fact return an allocation with maximum $p$-mean welfare.

We now claim that there is an allocation of $p$-mean welfare $\left(\frac{2}{m} \right)^{\frac{1}{p}} \frac{z}{2}$ iff the underlying \textsc{Partition} instance has the required set $T$. Suppose that there exists such a set $T$. Then, we can assign all goods with indices in the set $T$ to the first agent and the remaining goods to the second agent. The first two agents achieve value $z/2$ each under this allocation and, hence, the $p$-mean welfare is exactly $\left(\frac{2}{m} \right)^{\frac{1}{p}} \frac{z}{2}$. If the required set $T$ does not exist, then in any allocation one of the first two agents has value $y < z/2$, while the other has value at most $y' \le z-y$. All other agents have valuation $0$. Since the $p$-mean welfare function is strictly concave in the values for $p \in (0,1)$, it follows that in this case any allocation has $p$-mean welfare strictly less than  $\left(\frac{2}{m} \right)^{\frac{1}{p}} \frac{z}{2}$. Hence, the {\rm NP}-hardness of \textsc{Partition} implies that the fair division problem at hand is {\rm NP}-hard as well. 
\end{proof}

\end{document}